\pdfoutput=1
\documentclass[11pt]{article}
\usepackage{epsfig,amsmath,amssymb,mathtools,amsthm}
\usepackage{euscript}
\usepackage[authoryear,round]{natbib}

\usepackage[usenames,dvipsnames]{color}
\usepackage{pdflscape}
\usepackage{multirow}

\usepackage[table]{xcolor}
\usepackage{colortbl}

\makeatletter

\newcommand{\be}{\begin{eqnarray}}
\newcommand{\ee}{\end{eqnarray}}
\newcommand{\ben}{\begin{eqnarray*}}
\newcommand{\een}{\end{eqnarray*}}
\newcommand{\utwi}[1]{\mbox{\boldmath $ #1$}}

\newcommand{\ba}{{\utwi{a}}}
\newcommand{\bone}{{\utwi{1}}}

\newcommand{\bc}{{\utwi{c}}}

\newcommand{\bff}{{\utwi{f}}}
\newcommand{\bg}{{\utwi{g}}}
\newcommand{\bh}{{\utwi{h}}}

\newcommand{\bq}{{\utwi{q}}}
\newcommand{\br}{{\utwi{r}}}

\newcommand{\bx}{{\utwi{x}}}

\newcommand{\bA}{{\utwi{A}}}
\newcommand{\bB}{{\utwi{B}}}
\newcommand{\bC}{{\utwi{C}}}

\newcommand{\bE}{{\utwi{E}}}
\newcommand{\bF}{{\utwi{F}}}
\newcommand{\bG}{{\utwi{G}}}
\newcommand{\bH}{{\utwi{H}}}
\newcommand{\bI}{{\utwi{I}}}

\newcommand{\bM}{{\utwi{M}}}

\newcommand{\bO}{{\utwi{O}}}

\newcommand{\bQ}{{\utwi{Q}}}
\newcommand{\bR}{{\utwi{R}}}
\newcommand{\bS}{{\utwi{S}}}

\newcommand{\bU}{{\utwi{U}}}

\newcommand{\bW}{{\utwi{W}}}
\newcommand{\bX}{{\utwi{X}}}
\newcommand{\bY}{{\utwi{Y}}}
\newcommand{\bZ}{{\utwi{Z}}}

\newcommand{\bPhi}{{\utwi{\mathnormal\Phi}}}

\newcommand{\cD}{{\cal D}}
\newcommand{\cF}{{\cal F}}
\newcommand{\cM}{{\cal M}}

\newcommand{\bQone}{\bQ_1}
\newcommand{\bQtwo}{\bQ_2}

\newcommand{\bepsilon}{{\utwi{\epsilon}}}
\newcommand{\bGamma}{{\utwi{\Gamma}}}

\newcommand{\bOmega}{{\utwi{\Omega}}}
\renewcommand{\bPhi}{{\utwi{\Phi}}}

\newcommand{\bSigma}{{\utwi{\Sigma}}}

\newcommand{\widehatbQone}{\widehat{\bQ}_1}
\newcommand{\widehatbQtwo}{\widehat{\bQ}_2}

\def\Cov{\mathrm{Cov}}
\newcommand{\rmvec}{\mathrm{vec}}
\newcommand{\rmtr}{\mathrm{tr}}
\def\expec{\mathrm{E}}

\newcommand{\bzero}{{\utwi{0}}}

\newtheorem{theorem}{Theorem}
\newtheorem{lemma}{Lemma}

\begin{document}

\title{\bf Factor Models for Matrix-Valued High-Dimensional Time
  Series \footnote{
Chen's research was supported in part by National Science Foundation
grants DMS-1503409 and DMS-1209085.  Corresponding author: Rong Chen,
Department of Statistics, Rutgers University, Piscataway, NJ 08854,
USA. Email: rongchen@stat.rutgers.edu. } }

\author{Dong Wang \\ {\small Department of Operations Research and
    Financial Engineering, Princeton University, Princeton, NJ 08544} \\
Xialu Liu \\
{\small Management Information Systems Department, San Diego State
  University, San Diego, CA 92182} \\
Rong Chen \\
{\small Department of Statistics, Rutgers University, Piscataway, NJ 08854}
}
\date{June, 2017}
\maketitle

\begin{abstract}
In finance, economics and many other fields, observations in a matrix
form are often observed over time. For example, many economic
indicators are obtained in different countries over time. Various
financial characteristics of many companies are reported over
time. Although it is natural to turn a matrix observation into a long
vector then use standard vector time series models or factor analysis,
it is often the case that the columns and rows of a matrix represent
different sets of information that are closely interrelated in a very
structural way. We propose a novel factor model that maintains and
utilizes the matrix structure to achieve greater dimensional reduction
as well as finding clearer and more interpretable factor
structures. Estimation procedure and its theoretical properties are
investigated and demonstrated with simulated and real examples.
\end{abstract}

\section{Introduction}

Time series analysis is widely used in many applications. Univariate
time series, when one observes one variable through time, is well
studied, with linear models \citep[e.g.][]{Box&Jenkins76,
  Brockwell&Davis91, Tsay05}, nonlinear models \citep[e.g.][]{Engle82,
  bollerslev1986generalized,Tong90}, and nonparametric models
\citep[e.g.][]{Fan&Yao03}. Multivariate time series and panel time
series, when one observes a vector or a panel of variables through
time, is also a long studied but still active field \citep[e.g.][and
  others]{tiao1981, tiao1989, engle1995, Stock&Watson2004,
  lutkepohl2005, Tsay14}. Such analysis not only reveals the temporal
dynamics of the time series, but also explores the relationship among
a group of time series, using the available information more
fully. Often, the investigation of the relationship among the time
series is the objective of the study.

Matrix-valued time series, when one observes a group of variables
structured in a well defined matrix form over time, has not been
studied. Such a time series is encountered in many applications. For
example, in economics, countries routinely report a set of
economic indicators (e.g. GDP growth, unemployment rate, inflation
index and others) every quarter. Table~\ref{matrix} depicts such a
matrix-valued time series.  One can concentrate on one cell in
Table~\ref{matrix}, say US Unemployment rate series $\{X_{t,21},~t=1,2
\ldots\}$ and build a univariate time series model. Or one can concentrate
on one column in Table~\ref{matrix}, say, all economic indicators of
US $\{(X_{t,11},\ldots, X_{t,41})'\}$ and study it as a vector time
series. Similarly, if one is interested in modeling GDP growth of the
group of countries, a panel time series model can be built for the
first row $\{(X_{t,11},\ldots, X_{t,1p})\}$ in
Table~\ref{matrix}. However, there are certainly relationships among
all variables in the table and the matrix structure is extremely
important. For example, the variables in the same column (same
country) would have stronger inter-relationship. Same for the
variables in the same row (same indicator). Hence it is important to
analyze the entire group of variables while fully preserve and utilize
its matrix structure.

\begin{table} \label{matrix}
\begin{center}
\begin{tabular}{c|ccccc}
 & US & Japan & $\cdots$ & China \\ \hline
GDP & $X_{t,11}$ & $X_{t,12}$ & $\cdots$ & $X_{t,1p}$ \\
Unemployment & $X_{t,21}$ & $X_{t,22}$ & $\cdots$ & $X_{t,2p}$\\
Inflation & $X_{t,31}$ & $X_{t,32}$ & $\cdots$ & $X_{t,3p}$\\
Payout Ratio & $X_{t,41}$ & $X_{t,42}$ & $\cdots$ & $X_{t,4p}$\\
\end{tabular}
\caption{Illustration of a matrix-valued time series}
\end{center}
\end{table}


%

There are many other examples. Investors may be interested in the time
series of a group
of financials (e.g. asset/equity ratio, dividend per share, and
revenue) for a group of companies, the
import-export volume among a group of countries, pollution and
environmental variables (e.g. PM2.5, ozone level, temperature,
moisture, wind speed, etc) observed at a group of stations. In this
article we study such a matrix-valued time series.

Matrix-valued data has been studied
\citep[e.g.][]{gupta1999matrix, kollo2006advanced,
  werner2008estimation, leng2012sparse, yin2012model,
  zhao2014structured, zhou2014gemini, zhou2014regularized}. Their
study mainly focuses on independent observations. The concept of
matrix-valued time series was introduced by \citet{walden2001},
applied in signal and image processing. Still, the temporal dependence
of the time series was not fully exploited for model building.

In this article, we focus on high-dimensional matrix-valued time
series data. In cases, we may allow the dimensions of the matrix to be
as large as, or even larger than the length of the observations. A
well-known issue often accompanying with high-dimensional data is the
curse of dimensionality. We adopt a factor model approach.  Factor
analysis can effectively reduce the number of parameters involved, and
is a powerful statistical approach to extracting hidden driving
processes, or latent factor processes, from an observed stochastic
process. In the past decades, factor models for high-dimensional time
series data have drawn great attention from both econometricians and
statisticians \citep[e.g.][]{chamberlain1983arbitrage,
  forni2000generalized, bai2002determining, hallin2007determining,
  pan2008modelling, lam2011estimation, fan2011high, lam2012factor,
  fan2013large, chang2015high, liu2016regime}.

With the above observations and motivations, in this article, we
aim to develop factor models for matrix-valued time series, which
fully explore the matrix structure.  The rest of this article is
organized as follows. In Section 2, detailed model settings are
introduced and interpretations are discussed in detail. Section 3
presents an estimation procedure. The theoretical properties of the
estimators are also studied.  Simulation results are shown in Section
4 and two real data examples are given in Sections 5 and 6.
Section 7 provides a brief summary. All proofs are in
Appendix.

\section{Matrix Factor Models}

Let $\bX_t$ ($t=1,\ldots,T$) be a matrix-valued time series, where
each $\bX_t$ is a matrix of size $p_1 \times p_2$,
\[
\bX_t=\left( \begin{array}{ccc}
  X_{t,11} & \cdots & X_{t,1p_2} \\
 \vdots & \ddots & \vdots  \\
  X_{t, p_1 1} & \cdots & X_{t, p_1 p_2}
\end{array}\right).
\]
We propose the following factor model for matrix-valued time series,
\begin{equation}
\label{eqn:2dmodel}
\bX_t=\bR \bF_{t} \bC' + \bE_t, \quad t=1,2,\ldots,T.
\end{equation}
Here, $\bF_t$ is a $k_1\times k_2$ unobserved matrix-valued time
series of common {\bf fundamental factors}, $\bR$ is a $p_1 \times
k_1$ front loading matrix, $\bC$ is a $p_2 \times k_2$ back loading
matrix, and $\bE_t$ is a $p_1 \times p_2$ error matrix. In model
(\ref{eqn:2dmodel}), the common fundamental factors $\bF_t$'s drive all
dynamics and co-movement of $\bX_t$. $\bR$ and $\bC$ reflect the
importance of common factors and their interactions.

Similar to multivariate factor models, we assume that the
matrix-valued time series is driven by a few latent factors. Unlike
the classical factor model, the factors $\bF_t$'s in model
\eqref{eqn:2dmodel} are assumed to be organized in a matrix
form. Correspondingly, we adopt two loading matrices $\bR$ and $\bC$
to capture the dependency between each individual time series in the
matrix observations and the matrix factors. In the following we
provide two interpretations of the loading matrices. We first
introduce some notation. For a matrix $\bA$, we use $\ba_{i \cdot}$
and $\ba_{j}$ to represent the $i$-th row and the $j$-th column
of $\bA$, respectively, and $A_{ij}$ to denote the $ij$-th element of
$\bA$.

\vspace{0.2in}

\noindent
{\bf Interpretation I:} To isolate effects, assume $k_1=p_1$ and
$\bR=\bI_{p_1}$, then $\bX_t=\bF_{t}\bC' + \bE_t$. In this case, each
column of $\bX_t$ is a linear combination of the columns of
$\bF_{t}$. Take the example shown in Table~\ref{matrix} and consider the
first column of $\bX_t$ (the US economic indicators),

\vspace{.5in}

\centerline{
US \hspace{0.9in} $\bff_{t,1}$ \hspace{0.5in}
$\ldots$  \hspace{0.5in} $\bff_{t,k_2}$  \ \ \ \ \  \ \ \ \ \ \ \
}
\vspace{-0.05in}
\[
\left(\begin{array}{c}
\mbox{GDP} \\ \mbox{Unem} \\ \mbox{Inf} \\ \mbox{PayR}
\end{array}\right)_t
\!\!=\!C_{11}\left(\begin{array}{c}
\mbox{F-GDP} \\ \mbox{F-Unem} \\ \mbox{F-Inf} \\ \mbox{F-PayR}
\end{array}
\right)_{t}
\!\!\!\!+\cdots+
C_{1 k_2}\left(\begin{array}{c}
\mbox{F-GDP} \\ \mbox{F-Unem} \\ \mbox{F-Inf} \\ \mbox{F-PayR}
\end{array}
\right)_{t}+\mathbf{e}_{t,US} .
\]
It is seen that the US GDP only depends on the first row of $\bF_t$.
Similarly, other countries' GDP also only depends on the first row of
$\bF_t$. Hence we can view the first row of $\bF_t$ as the GDP
factors. Similarly, the second row of $\bF_t$ can be considered as the
unemployment factors.  There is no interaction between the indicators
in this setting (when $\bR=\bI$).  The loading matrix $\bC$ reflects
how each country (column of $\bX_t$) depends on the columns of
$\bF_t$, hence reflects column interactions, or the interactions
between the countries.  Because of this, we will call $\bC$ the
column loading matrix.

Similarly, the rows of $\bF_t$ can be viewed as common factors of all
rows of $\bX_t$, and the front loading matrix $\bR$ as row loading
matrix. Again, assume $k_2=p_2$ and $\bC=\bI_{p_2}$, it follows that
$\bX_t=\bR\bF_{t}+\bE_t$. Then each row of $\bX_t$ is a linear
combination of the rows of $\bF_{t}$. Consider the first row of
$\bX_t$,

\vspace{.2in}

\centerline{
\hspace{0.15in} US \hspace{0.05in} Japan \hspace{0.005in}...\hspace{0.005in} China
\hspace{0.7in} US \hspace{0.18in}
Japan\hspace{0.15in}...\hspace{0.15in}China \ \ \ \ \ \ \ \ \ \ \ \ \ \ \ \  \ }
\vspace{-0.3in}
\begin{eqnarray*}
(\mbox{GDP}, \mbox{GDP}, \ldots, \mbox{GDP})_{t} = &
R_{11}(\mbox{F-US}, \mbox{F-Japan}, \ldots, \mbox{F-China})_{t} & \bff_{t,1\cdot}\\
 & + R_{12}
(\mbox{F-US}, \mbox{F-Japan}, \ldots, \mbox{F-China})_{t}  & \bff_{t,2\cdot}\\
 & + \cdots  & \vdots \\
 & + R_{1k_1}
(\mbox{F-US}, \mbox{F-Japan}, \ldots, \mbox{F-China})_{t} &  \bff_{t,k_1\cdot} \\
 & + \mathbf{e}_{t,GDP\cdot}. &
\end{eqnarray*}
It is seen that all economic movements (of each country) are driven by
$k_1$ (row) common factors. For example, every US's indicator depends
on only the first column of $\bF_t$. Hence the first column of $\bF_t$
can be viewed as the US factor. And the second column of $\bF_t$ can
be viewed as Japan factor. The loading matrix $\bR$ reflects how each
indicator depends on the rows of $\bF_t$. It reflects row
interactions, the interactions between the indicators within each
country. Because of this, we will call $\bR$ the row loading matrix.


Obviously column and row interaction would be of interests and of
importance. One way to introduce interaction is to assume an additive
structure, by combining the column and row factor models
\[
\bX_t=\bR \bF_{1t} + \bF_{2t}\bC'+ \bE_t, \quad t=1,2,\ldots,T.
\]
However, the number of factors in this model is large
($k_1\times p_2+p_1\times k_2$). A more parsimonious  model would be a direct
interaction as in model (\ref{eqn:2dmodel}). In this case the number of
factors is only $k_1\times k_2$.

\vspace{0.2in}

\noindent
{\bf Interpretation II:} We can view the model (\ref{eqn:2dmodel})
as a two-step
hierarchical model.

\noindent{\bf Step 1:} For each fixed row $i = 1,2, \ldots, p_1$, using
data $\{\bx_{t,i\cdot},~t=1,2,\ldots,T\}$, we can find a $p_2 \times
k_2$ dimensional loading matrix $\bC^{(i)}$ and $k_2$ dimensional factors
$\{\bg_{t,i\cdot}=(G_{t,i1},\dots, G_{t,ik_2}),~t=1,2,\ldots,T\}$ under a standard vector factor model
setting. That is,
\[
(X_{t,i1},\ldots,X_{t,i p_2})=(G_{t,i1},\ldots,G_{t,ik_2}) \bC^{(i)'} +
(H_{t,i1},\ldots,H_{t,i p_2}), \quad t=1,2,\ldots,T.
\]
Let $\bG_t$ be the $p_1 \times k_2$ matrix formed with $p_1$ rows of $\bg_{t,i\cdot}$. Also denote $\bH_t$ as the $p_1\times p_2$ error matrix
formed with the rows of $\{\bh_{t,i\cdot}= (H_{t,i1} \ldots, H_{t,ip_2})\}$.

\noindent{\bf Step 2:} Suppose each column $j=1,2,\ldots,k_2$ of the
assembled factor matrix $\bG_t$ obtained in Step 1 also assumes the
factor structure, with a $p_1\times k_1$ loading matrix $\bR^{(j)}$ and
a $k_1$ dimensional factor $\bff_{t,j}$. That is,
\[
\left( \begin{array}{c} G_{t,1j} \\ \vdots \\ G_{t,p_1 j}\end{array} \right)
=\bR^{(j)} \left( \begin{array}{c}  F_{t,1j} \\ \vdots \\ F_{t,k_1j} \end{array} \right)
+
 \left( \begin{array}{c}  H^*_{t,1j} \\ \vdots \\ H^*_{t,p_1j} \end{array} \right),
\quad t=1,2,\ldots,T.
\]
This step reveals the common factors that drive the co-moments in
$\bG_t$. Let $\bF_t$ be the $k_1 \times k_2$ matrix
formed with the columns $\bff_{t, j}$. And let
$\bH^*_t$ be the $p_1 \times k_2$ error matrix formed with columns
$\{\bh^*_{t,j}= (H^*_{t,1j}, \ldots, H^*_{t,p_1j} )'\}$.

\noindent{\bf Step 3: Assembly:} With the above two-step factor
analysis and notation, assume $\bR^{(1)}=\ldots=\bR^{(k_2)}=\bR$ and
$\bC^{(1)}=\ldots=\bC^{(p_1)}=\bC$, we have
\[
\bX_t=\bG_t\bC^{'}+\bH_t \mbox{\ \ and \ \ } \bG_t=\bR\bF_t+\bH^*_t.
\]
Hence
\[
\bX_t=\bR \bF_t \bC^{'}+ \bH_t^* \bC^{'} + \bH_t = \bR \bF_t \bC^{'} + \bE_t,
\]
where $\bE_t=\bH_t^* \bC^{'} + \bH_t$. It is identical to
\eqref{eqn:2dmodel}.

\medskip

Here we provide some additional remarks of  model
\eqref{eqn:2dmodel}.

\noindent
{\bf Remark 1:\ } Let $\rmvec(\cdot)$ be the vectorization operator,
i.e., $\rmvec(\cdot)$ converts a matrix to a vector by stacking
columns of the matrix on top of each other. The classical factor
analysis treats $\rmvec(\bX_t)$ as the observations, and a factor
model is in the form of
\begin{equation}
\label{eqn:1dmodel}
\rmvec(\bX_t) = \bPhi {\bf f}_t + {\bf e}_t, \quad t=1,2,\ldots,T,
\end{equation}
where $\bPhi$ is a $p_1p_2 \times k$ loading matrix, ${\bf f}_t$ of
length $k$ is the latent factor, ${\bf e}_t$ is the error term, and
$k$ is the total number of factors. On the other hand, note that model
\eqref{eqn:2dmodel} can be re-written as
\begin{equation}
\label{eqn:2dmodel1d}
\rmvec(\bX_t)=(\bC \otimes \bR)
\rmvec(\bF_{t})+\rmvec(\bE_t).
\end{equation}
Assume $k=k_1k_2$. Then model \eqref{eqn:2dmodel1d} is a special case of
model \eqref{eqn:1dmodel}, with a Kronecker product structured
loading matrix. Hence model \eqref{eqn:2dmodel} is a restricted
version of model \eqref{eqn:1dmodel}, assuming a special structure for
the loading spaces. The number of parameters for the loading matrix
$\bPhi$ in model \eqref{eqn:1dmodel} is $(p_1k_1) \times (p_2k_2)$
whereas it is $p_1k_1 + p_2k_2$ for the loading matrices $\bR$ and $\bC$
in model \eqref{eqn:2dmodel}. Therefore, model \eqref{eqn:2dmodel}
significantly reduces the dimension of the problem.

\medskip

\noindent
{\bf Remark 2:\ } Interpretation II also reveals the reduction in
the number of factors comparing to using factor models for each
column panel or row panel. Note that, if one ignores the interconnection
between the rows and obtain individual
factor models for each row, as in Step 1,
the total number of factors is $p_1\times k_2$. These factors may have
connections across rows. Step 2 exploits such correlations and uses another
factor model to reduce the number of factors from $p_1\times k_2$ to
$k_1\times k_2$.

\medskip

\noindent
{\bf Remark 3:\ }
We also observed that in practice, the total number of factors used in
our model may be larger than the number of factors needed in the
vectorized factor model (\ref{eqn:1dmodel}).
This is possible since
the vectorized model simultaneously exploits common driving features
in all series, while the matrix factor model does it by working on the row
vectors separately first (Step 1), then condensing them by the columns
(Step 2). Such a two-step approach may result in redundancy (highly
correlated factors) which may be
further simplified. Because we are forcing the
factors to assume a neat matrix structure, it is difficult to have
simplifications such as having
one or several elements in the factor matrix $\bF_t$ be constant zero. Since
$k_1$ and $k_2$ are usually small, we will tolerate such redundancy.
One extension is to assume that
the factor matrix $\bF_t$, after a certain rotation, has a block diagonal
structure, resulting in a multi-term factor model
\begin{align}
\bX_t=\sum_{i=1}^s \bR_i \bF_{it} \bC_i^{'} + \bE_t, \quad t=1,2,\ldots,T, \label{multi}
\end{align}
where $\bF_{it}$ is a $k_{i1}\times k_{i2}$ factor matrix, and
$\sum_{i=1}^s k_{i1}=k_1$ and $\sum_{i=1}^s k_{i2}=k_2$. This will reduce
the number of factors from $k_1\times k_2$ to $\sum_{i=1}^s k_{i1}k_{i2}$,
with corresponding dimension reduction in the loading matrices as well.
We are currently investigating the properties and estimation procedures of
such a  multi-term factor model.

\medskip

\noindent
{\bf Remark 4:\ }
As in all factor model setting, the properties or assumptions on
the observed process $\bX_t$ are inferred from the assumptions on the
factors and the noise processes, since the observed series are assumed to be
linear combinations of the factor processes plus the noise process. Indirectly, we assume that
all autocovariance matrices of lag $h\geq 1$ of all series lie in a
structured $k_1k_2\times k_1k_2$ space, but no assumption on the
contemporary covariance matrix, as we do not assume any
contemporary covariance structure on the error $\bE_t$.

\medskip

\noindent
{\bf Remark 5:\ } Similar models as model \eqref{eqn:2dmodel} have
been proposed and studied when conducting principal component analysis
on matrix-valued data \citep[e.g.][]{paatero1994fac, yang2004two,
  ye2005generalized, ding2005twod, zhang20052d,
  crainiceanu2011population, wang2016efficient}. In those studies, the
matrix-valued observations $\bX_t$ are assumed to be independent, and
they primarily focused on principal component analysis. To the best of
our knowledge, our paper is the first one considering factor models
for matrix-valued time series data.

In this article, we extend the methods described in
\citet{lam2011estimation} and \citet{lam2012factor} for vector-valued
factor model \eqref{eqn:1dmodel} to matrix-valued factor model
\eqref{eqn:2dmodel}. We propose estimators for the loading spaces and
the numbers of row and column factors, investigate their theoretical
properties, and establish their convergence rates. Simulated and
real examples are presented to illustrate the performance of the
proposed estimators, to compare the asymptotics under different
conditions with different factor strengths, and to explore
interactions between row and column factors.

\section{Estimation and Modeling Procedures}
\label{sec:method}

Because of the latent nature of the factors, various assumptions are
imposed to `define' a factor. Two common assumptions are used. One
assumes that the factors must have impact on most of the series, and weak
serial dependence is allowed for the idiosyncratic noise process, see
\citet{chamberlain1983arbitrage, forni2000generalized,
  bai2002determining, hallin2007determining}, among others.  Another
assumes that the factors should capture all dynamics of the observed
process, hence the idiosyncratic noise process has no serial
dependence (but may have strong cross-sectional dependence), see
\citet{pan2008modelling, lam2011estimation, lam2012factor,
  chang2015high, liu2016regime}.  Here we adopt the second assumption
and assume that the vectorized error $\rmvec(\bE_t)$ is a white noise
process with mean $\bzero$ and covariance matrix $\bSigma_e$, and is
independent of the factor process $\rmvec(\bF_t)$. For ease of
presentation, we will assume that the process $\bF_t$ has mean
$\bzero$, and the observations $\bX_t$'s are centered and standardized
through out this paper.

For the vector-valued factor model \eqref{eqn:1dmodel}, it is
well-known that there exists an identifiable issue among the factors
${\bf f}_t$ and the loading matrix $\bPhi$. Similar problem also
arises in the proposed matrix-valued factor model
\eqref{eqn:2dmodel}. Let $\bU_1$ and $\bU_2$ be two invertible
matrices of sizes $k_1 \times k_1$ and $k_2 \times k_2$. Then the
triplets $(\bR,\bF_t,\bC)$ and $(\bR \bU_1 ,\bU_1^{-1}\bF_t\bU_2^{-1},
\bC \bU_2')$ are equivalent under model \eqref{eqn:2dmodel}, and hence
model \eqref{eqn:2dmodel} is not identifiable. However, with a similar
argument as in \citet{lam2011estimation} and \citet{lam2012factor},
the column spaces of the loading matrices $\bR$ and $\bC$ are uniquely
determined. Hence, in the following, we will focus on the estimation
of the column spaces of $\bR$ and $\bC$, denoted by $\cM(\bR)$ and
$\cM(\bC)$, and referred to as row factor loading space and column
factor loading space, respectively.

We can further decompose $\bR$ and $\bC$ as
follows,
\begin{equation*}
\label{eqn:qrmodel}
\bR = \bQ_1 \bW_1, \mbox{ and }\bC= \bQ_2 \bW_2,
\end{equation*}
where $\bQ_i$ is a $p_i \times k_i$ matrix with orthonormal columns
and $\bW_i$ is a $k_i \times k_i$ non-singular matrix, for
$i=1,2$. Let $\cM(\bQ_i)$ denote the column space of $\bQ_i$. Then we
have $\cM(\bQ_1) = \cM(\bR)$ and $\cM(\bQ_2) = \cM(\bC)$ . Hence, the
estimation of column spaces of $\bR$ and $\bC$ is equivalent to the
estimation of column spaces of $\bQ_1$ and $\bQ_2$.

Write
\begin{equation*}
\bZ_t = \bW_1 \bF_t \bW_2', \quad t=1,2,\ldots,T,
\end{equation*}
as a transformed latent factor process. Then, model
\eqref{eqn:2dmodel} can be re-expressed as
\begin{equation}
\label{eqn:2dmodelqr}
\bX_t = \bQ_1 \bZ_t \bQ_2' + \bE_t, \quad
t=1,2, \ldots,T.
\end{equation}
Equation \eqref{eqn:2dmodelqr} can be viewed as another formulation of
the matrix-valued factor model with orthonormal loading
matrices. Since $\cM(\bR) = \cM(\bQ_1)$ and $\cM(\bC)=\cM(\bQ_2)$, we
will perform analysis on model \eqref{eqn:2dmodel} and
\eqref{eqn:2dmodelqr} interchangeably whenever one is more convenient
than the other.

\subsection{Estimation}
\label{subsec:est}

To estimate the matrix-valued factor model \eqref{eqn:2dmodel}, we
follow closely the idea of \citet{lam2011estimation} and
\citet{lam2012factor} in estimating vector-valued factor models.  The
key idea is to calculate auto-cross-covariances of the time series
then construct a Box-Ljung type of statistics in matrix.  Under the
matrix factor model and white idiosyncratic noise assumption, the
space spanned by such a matrix is directly linked with the loading
matrices.  In what follows, we will illustrate the method to obtain an
estimate of $\cM(\bR)$. The column space of $\bC$ can be estimated in a
similar way using the transposes of $\bX_t$'s.

Let the $j$-th column of $\bX_t$, $\bR$, $\bC$, $\bQ_i$ and $\bE_t$ be
$\bx_{t,j}$, $\br_{j}$, $\bc_{j}$, $\bq_{i,j}$ and $\bepsilon_{t,j}$,
respectively.  Let $\br_{k\cdot}$, $\bc_{k\cdot}$ and $\bq_{i,k\cdot}$
be the row vectors that denote the $k$-th row of $\bR$, $\bC$ and
$\bQ_i$, respectively. Then it follows from \eqref{eqn:2dmodel} and
\eqref{eqn:2dmodelqr} that
\begin{equation}
\label{eqn:ytj}
\bx_{t,j} = \bR \bF_t \bc_{j\cdot}' + \bepsilon_{t, j} = \bQone \bZ_t \bq_{2,j\cdot}' + \bepsilon_{t, j}, \quad j = 1, 2,
\ldots, p_2.
\end{equation}
From the zero mean assumptions of both $\bF_t$ and $\bE_t$, we have
$\expec(\bx_{t,j}) = \bzero$.

Let $h$ be a positive integer. Define
\begin{eqnarray}
\label{eqn:omegafqijdef}
\bOmega_{zq,ij}(h) &=& \frac{1}{T-h} \sum_{t=1}^{T-h} \Cov(\bZ_t
\bq_{2,i\cdot}', \bZ_{t+h} \bq_{2,j\cdot}'), \\
\label{eqn:omegayijdef}
\bOmega_{x,ij}(h) &=& \frac{1}{T-h} \sum_{t=1}^{T-h} \Cov(\bx_{t,
  i},\bx_{t+h,j}),
\end{eqnarray}
for $i,j=1,2,\ldots,p_2$. By plugging \eqref{eqn:ytj} into
\eqref{eqn:omegayijdef} and by the assumption that $\bE_t$ is white,
 it follows that
\begin{equation}
\label{eqn:omegayijsimpleform}
\bOmega_{x,ij}(h) = \bQone \bOmega_{zq,ij}(h) \bQone',
\end{equation}
for $h\geq1$. For a pre-determined integer $h_0$, define
\begin{equation}
\label{eqn:Mdef}
\bM_1 = \sum_{h=1}^{h_0} \sum_{i=1}^{p_2} \sum_{j=1}^{p_2}
\bOmega_{x,ij}(h) \bOmega_{x,ij}'(h).
\end{equation}

\noindent
By \eqref{eqn:omegayijsimpleform} and \eqref{eqn:Mdef}, it
follows that
\begin{equation}
\label{eqn:Mexpression}
\bM_1 = \bQone \left(\sum_{h=1}^{h_0} \sum_{i=1}^{p_2} \sum_{j=1}^{p_2}
\bOmega_{zq,ij}(h) \bOmega_{zq,ij}'(h) \right) \bQone'.
\end{equation}

Suppose the matrix $\bM_1$ has rank $k_1$ (Condition 5 in Section
\ref{sec:thm}).
From \eqref{eqn:Mexpression}, we can see that
each column of $\bM_1$ is a linear combination of columns of $\bQone$,
and thus the matrices $\bM_1$ and $\bQone$ have the same column spaces,
that is, $\cM(\bM_1) = \cM(\bQone)$. It follows that the eigen-space of
$\bM_1$ is the same as $\cM(\bQone)$. Hence, $\cM(\bQone)$ can be
estimated by the space spanned by the eigenvectors of the sample
version of $\bM_1$.
Assume that $\bM_1$ has $k_1$ distinct nonzero eigenvalues, and let $\bq_{1,j}$
be the unit eigenvector corresponding to the $j$-th largest eigenvalue. As
there are two unit eigenvectors corresponding to each eigenvalue, we use
the one with positive $\bone'\bq_{1,j}$.
We can now uniquely define $\bQ_1$ by
\[
\bQ_1=( \bq_{1,1}, \bq_{1,2}, \ldots, \bq_{1,k_1}).
\]

Now we construct the sample versions of these quantities and introduce
the estimation procedure as follows. For any positive integer $h$ and
a pre-scribed positive integer $h_0$, let
\begin{eqnarray}
\label{eqn:hatomegayijdef}
\widehat{\bOmega}_{x,ij}(h) &=& \frac{1}{T-h} \sum_{t=1}^{T-h} \bx_{t,i}
\bx_{t+h, j}', \\
\label{eqn:hatmdef}
\widehat{\bM}_1 &=& \sum_{h=1}^{h_0} \sum_{i=1}^{p_2} \sum_{j=1}^{p_2}
\widehat{\bOmega}_{x,ij}(h) \widehat{\bOmega}_{x,ij}'(h).
\end{eqnarray}
Then, $\cM(\bQ_1)$ can be estimated by $\cM(\widehat{\bQ}_1)$, where
$\widehat{\bQ}_1=\{\widehat{\bq}_{1,1}, \ldots, \widehat{\bq}_{1,k_1}\}$,
and $\widehat{\bq}_{1,1}, \ldots \widehat{\bq}_{1,k_1}$ are the eigenvectors
of $\widehat{\bM}_1$ corresponding to its $k_1$ largest eigenvalues.

In practice, the number of row factors $k_1$ is usually
unknown. This quantity can be estimated through a similar eigenvalue
ratio estimator as described in \citet{lam2012factor}. Let
$\widehat{\lambda}_{1,1} \geq \widehat{\lambda}_{1,2} \geq \ldots \geq
\widehat{\lambda}_{1,p_1} \geq 0$ be the ordered eigenvalues of
$\widehat{\bM}_1$. Then
\[
\widehat{k}_1 = \mathrm{arg\,min}_{1 \leq i \leq p_1/2}
\frac{\widehat{\lambda}_{1,i+1}}{ \widehat{\lambda}_{1,i}} .
\]

For $\bQtwo$ and $k_2$, they can be estimated by performing the same
procedure on the transposes of $\bX_t$'s to construct $\bM_2$ and $\widehat{\bM}_2$. Once $\widehatbQone$ and
$\widehatbQtwo$ are obtained, the estimate of $\bZ_t$ can be found via
a general linear regression analysis, since
\[
\rmvec(\bX_t)=(\bQtwo \otimes \bQone)
\rmvec(\bZ_{t})+\rmvec(\bE_t).
\]
Together with the orthonormal properties of both $\widehat{\bQ}_1$ and
$\widehat{\bQ}_2$ and the properties of Kronecker product, it follows that
\[
\widehat{\bZ}_t = \widehatbQone' \bX_t \widehatbQtwo.
\]
Let $\bS_t$ be the dynamic signal part of $\bX_t$, that is,
$\bS_t=\bR \bF_t \bC'=\bQ_1 \bZ_t \bQ_2'$. Then a natural
estimator of $\bS_t$ is given by,
\begin{equation}\label{S}
\widehat{\bS}_t = \widehat{\bQ}_1 \widehat{\bQ}_1' \bX_t \widehat{\bQ}_2 \widehat{\bQ}_2'.
\end{equation}
\noindent{\bf Remark 6:} Theoretically any $h_0$ can be used to estimate the
loading spaces, as long as one of the $\bOmega_{x,ij}(h)$ is of full rank for $i,j=1,\ldots, p_2$ and $h=1,\ldots, h_0$. Although they converge at the same rate, the estimate from the lag where the autocorrelation maximizes is most efficient. We demonstrate the
impact of $h_0$ in the matrix factor model setting in Section 4.
As the autocorrelation is often at its strongest at small time lags, a relatively small $h_0$ is usually adopted \citep{lam2011estimation,chang2015high,liu2016regime}. Larger $h_0$ strengthens the signal, but also adds more noises
in the estimation of $\bM_i$.

\noindent{\bf Remark 7:} $K$-fold
cross-validation procedures can be adopted for model selection between matrix-valued factor models in (\ref{eqn:2dmodel}) and vector-valued factor models in (\ref{eqn:1dmodel}), and among the models with
different number of factors. Specifically, we first partition the data $D$
into $k$ subsets $D_1, \ldots, D_k$, and fit a factor model
with each of the $D\backslash D_k$ sets.
 Then we use the estimated loading spaces, together with the data in $D_k$ to
obtain the dynamic signal process $\bS_t$ for $D_k$, and obtain out-of-sample
residuals. Residual sum of squares (RSS) of the $K$ folds is then adopted for
model comparison. Rolling-validation which uses only the data before the
block for estimation can be used as well.

\subsection{Theoretical Properties of the Estimator}
\label{sec:thm}

In this section, we study the asymptotic properties of the estimators
under the setting that all $T$, $p_1$ and $p_2$ grow to infinity while
$k_1$ and $k_2$ are being fixed. In the following, for any matrix
$\bY$, we use ${\rm rank} (\bY)$, $\| \bY \|_2$, $\| \bY \|_F$, $\| \bY \|_{\min}$, and $\sigma_j(\bY)$ to
denote the rank, the spectral norm, the Frobenius norm, the smallest nonzero
singular value and the $j$-th largest singular value of $\bY$. When $\bY$ is a square matrix, we denote by
$\rmtr(\bY)$, $\lambda_{\max}(\bY)$ and $\lambda_{\min}(\bY)$ the
trace, maximum and minimum eigenvalues of $\bY$, respectively. We write
$a \asymp b$ when $a = O(b)$ and $b = O(a)$. Define
\begin{equation*}
\bSigma_{f}(h) = \frac{1}{T-h}\sum_{t=1}^{T-h} \Cov \left(\rmvec
  (\bF_t), \rmvec(\bF_{t+h}) \right), \ \ \
\mbox{and} \ \ \
\bSigma_e = {\rm Var}(\rmvec(\bE_t)).
\end{equation*}

The following regularity conditions are imposed before we derive the
asymptotics of the estimators.

\vspace{.1in}
\noindent {\bf Condition 1.} The vector-valued process $\rmvec(\bF_t)$
is $\alpha$-mixing. Specifically, for
some $\gamma > 2$, the mixing coefficients
satisfy the condition $\sum_{h=1}^\infty \alpha(h)^{1-2/\gamma} <
\infty$, where
\[
\alpha(h) = \sup_i \sup_{A \in {\cF}^i_{-\infty}, B \in
  {\cF}^{\infty}_{i+h}} | P(A \cap B) - P(A) P(B) |,
\]
and ${\cF}_i^j$ is the $\sigma$-field generated by $\{\rmvec(\bF_t):i
\leq t \leq j\}$.
\vspace{.1in}

\noindent {\bf Condition 2.} Let $F_{t,ij}$ be the $ij$-th entry of
$\bF_t$. For any $i=1,\ldots, k_1$, $j=1, \ldots, k_2$, and
$t=1,\ldots, T$, we assume that $\expec (|F_{t,ij}|^{2\gamma})\leq C$, where
$C$ is a positive constant, and $\gamma$ is given in Condition
1. In addition, there exists an $1 \leq h \leq h_0$ such that
${\rm rank}(\bSigma_f(h)) \geq k$, and $\| \bSigma_f(h)\|_2 \asymp O(1) \asymp \sigma_k(
\bSigma_f(h))$, where $k=\max\{k_1, k_2\}$, as $p_1$ and $p_2$ go to infinity and $k_1$ and $k_2$ are
fixed. For $i=1,\ldots, k_1$ and $j=1,\ldots, k_2$, $\frac{1}{T-h}\sum_{t=1}^{T-h}\Cov(\bff_{t,i}, \bff_{t+h,i}) \neq \mathbf{0}$, $\frac{1}{T-h}\sum_{t=1}^{T-h}\Cov(\bff_{t,j\cdot}, \bff_{t+h,j\cdot}) \neq \mathbf{0}$.

\vspace{.1in}
The latent process does not have to be stationary, but needs to
satisfy the mixing condition (Condition 1) and boundedness condition
(Condition 2). They are weaker than stationarity. For example,
when a process has a deterministic seasonal variance component or with
a deterministic regime switching mechanism, it
is not stationary but mixing.
We do not need to assume any specific
model for the latent process $\{\bF_t\}$ since we only use the
eigen-analysis based on autocovariances of the observed process at
nonzero lags.

Under Condition 2, $\bSigma_f(h)$ may not be of full rank, which indicates that it is allowed to involve some extent of redundancy in the factors. Condition 2 also guarantees that there is no redundant row or column in $\bF_t$, and in each row or column there is at least one factor which has serial dependence at lag $h$. The greater dimension reduction can be achieved by a multi-term factor model in (\ref{multi}). We are currently investigating the properties and estimation procedures of
such a multi-term factor model.

\vspace{.1in}
\noindent {\bf Condition 3.}  Each element of $\bSigma_e$ remains
bounded as $p_1$ and $p_2$ increase to infinity.
\vspace{.1in}

In model \eqref{eqn:2dmodel}, $\bR \bF_t \bC'$ can be viewed as
the signal part of the observation $\bX_t$, and $\bE_t$ as the
noise. The signal strength, or the strength of the factors, can be
measured by the $L_2$-norm of the loading matrices which are assumed
to grow with the dimensions.

\vspace{.1in}
\noindent {\bf Condition 4.} There exist constants $\delta_1$ and
$\delta_2 \in [0,1]$ such that $\|\bR\|_2^2 \asymp p_1^{1-\delta_1}
\asymp \|\bR\|_{\min}^2$ and $\|\bC\|_2^2 \asymp p_2^{1-\delta_2}
\asymp \|\bC\|_{\min}^2$, as $p_1$ and $p_2$ go to infinity and $k_1$
and $k_2$ are fixed.

\vspace{.1in}

The rates $\delta_1$ and $\delta_2$ are called the strength for row
factors and the strength for column factors, respectively. They measure
the relative growth rate of the amount of information carried by the observed
process $\bX_t$ on the common factors as the dimensions
increase, with respect to the growth rate of the amount of noise. When
$\delta_i=0$, the factors are strong; when $\delta_i >0$, the factors
are weak, which means the information contained in $\bX_t$ on the
factors grows more slowly than the noises introduced as $p_i$
increases. For detailed discussion of factor strength, see
\cite{lam2012factor}.

\vspace{.1in}
\noindent {\bf Condition 5.} $\bM_i$ has $k_i$ distinct positive
eigenvalues for $i=1,2$.
\vspace{.1in}

As stated in Section \ref{sec:method}, only $\cM(\bQ_1)$ and
$\cM(\bQ_2)$ are uniquely determined, while $\bQ_1$ and $\bQ_2$ are
not. However, when the eigenvalues of $\bM_i$ are distinct, we can
uniquely define $\bQ_i$ as $\bQ_i=\{\bq_{i,1}, \cdots, \bq_{i,k_i}\}$, where
$\bq_{i,1}, \cdots, \bq_{i,k_i}$ are the unit eigenvectors of $\bM_i$ corresponding to
its $k_i$ largest eigenvalues $\{\lambda_{i,1} > \lambda_{i,2} \ldots >
\lambda_{i,k_i}\}$ which make $ \mathbf{1}' \bq_{i,1}$, $\bone'\bq_{i,2}$, $\ldots$, and $\bone'\bq_{i,k_i}$ all positive, for $i=1,2$.

The following theorems show the rate of convergence for estimators of
loading spaces and the eigenvalues.
\begin{theorem}
\label{thm:qrate}
Under Conditions 1-5 and $p_1^{\delta_1} p_2^{\delta_2} T^{-1/2} =
o(1)$, it holds that
\begin{equation*}
\| \widehat{\bQ}_i - \bQ_i \|_2 = O_p(p_1^{\delta_1} p_2^{\delta_2}
T^{-1/2}), \mbox{ for } i=1,2.
\end{equation*}
\end{theorem}

Concerning the impact of $\delta_i$'s, it is not surprising that the stronger
the factors are, the more useful information the observed process
carries and the faster the estimators converge.  More interestingly,
the strengths of row factors and column factors $\delta_1$ and
$\delta_2$ determine the rates together.  An increase in the strength
of row factors is able to improve the estimation of the column factors
loading space and vice versa.

When $p_1$ and $p_2$ are fixed, the convergence rate for estimating the
loading matrices are $\sqrt{T}$. If the loadings are strong ($\delta_i=0$),
the rate is also $\sqrt{T}$, since the signal is as strong as the noise,
and the increase in dimensions will
not affect the estimation of the loading spaces.
When $\delta_i$'s are not 0, the  noise increases faster than
useful information. In this case, increases in dimension
will dilute the information, resulting in less efficient estimators.

\begin{theorem}
\label{thm:eigval}
With Conditions 1-5 and
$p_1^{\delta_1} p_2^{\delta_2} T^{-1/2} = o(1)$, the eigenvalues $\{\widehat{\lambda}_{i,1}, \ldots, \widehat{\lambda}_{i,p_i}\}$ of $\widehat{\bM}_i$  which are sorted in descending order satisfy
\begin{eqnarray*}
|\widehat{\lambda}_{i,j} - \lambda_{i,j}| &=& O_p (p_1^{2-\delta_1} p_2^{2-\delta_2}
T^{-1/2}), \quad \mbox{for  } j=1,2,\ldots,k_i, \\
\mbox{and} \qquad \qquad \qquad |\widehat{\lambda}_{i,j}| &=& O_p (p_1^{2} p_2^{2}
T^{-1}), \quad \mbox{for  } j=k_i+1,\ldots,p_i,
\end{eqnarray*}
where $\lambda_{i,1} >\lambda_{i,2} \ldots >\lambda_{i,k_i}$ are eigenvalues of
$\bM_i$, for $i=1,2$.
\end{theorem}

Theorem \ref{thm:eigval} shows that the estimators for nonzero
eigenvalues of $\bM_i$ converge more slowly than those for the zero
eigenvalues. It provides the theoretical support for the ratio
estimator proposed in Section \ref{subsec:est}. 

The following theorem demonstrates the theoretical properties of
the estimator $\widehat{\bS}_t$ in (\ref{S}).

\begin{theorem} If Conditions 1-5 hold, $p_1^{\delta_1} p_2^{\delta_2} T^{-1/2} = o(1)$, and $\|\bSigma_e \|_2$ is bounded, we have
\begin{eqnarray*}
p_1^{-1/2}p_2^{-1/2} \| \widehat{\bS}_t-\bS_t\|_2 &=&O_p( a_1\|\widehat{\bQ}_1 -{\bQ}_1\|_2) +O_p( a_2\|\widehat{\bQ}_2 -{\bQ}_2\|_2) +O_p(p_1^{-1/2}p_2^{-1/2})\\
&=&O_p (p_1^{\delta_1/2}
p_2^{\delta_2/2} T^{-1/2} +p_1^{-1/2}p_2^{-1/2}),
\end{eqnarray*}
where $a_1 \asymp a_2 \asymp  O(p_1^{-\delta_1/2} p_2^{-\delta_2/2}T^{-1/2})$.
\end{theorem}

The theorem shows that, in order to estimate the signal $\bS_t$ consistently,
dimensions $p_1$ and $p_2$ must go to infinity, in order to have sufficient
information on $\bS_t$ at  each time point $t$.

Since $\bQ_i$ is not identifiable in model \eqref{eqn:2dmodel},
another measure to quantify the accuracy of factor loading matrices
estimation is the distance between $\cM(\bQ_i)$ and
$\cM(\widehat{\bQ}_i)$. For two orthogonal matrices $\bO_1$ and
$\bO_2$ of sizes $p \times q_1$ and $p \times q_2$, define
\begin{equation*}
\cD(\bO_1,\bO_2) = \left(1 - \frac{1}{\max(q_1,q_2)} \rmtr(\bO_1 \bO_1'
\bO_2 \bO_2') \right)^{1/2}.
\end{equation*}
Then $\cD(\bO_1,\bO_2)$ is a quantity between $0$ and $1$. It is equal
to $0$ if the column spaces of $\bO_1$ and $\bO_2$ are the same and
$1$ if they are orthogonal.

\begin{theorem}
If Conditions 1-5 hold and $p_1^{\delta_1}p_2^{\delta_2} T^{-1/2} =
o(1)$, we have
\begin{equation*}
\cD (\widehat{\bQ}_i, \bQ_i)= O_p(p_1^{\delta_1}p_2^{\delta_2} T^{-1/2}), \, \mbox{ for } i=1,2.
\end{equation*}
\end{theorem}
Theorem 4 shows that the error to estimate loading spaces is on the same order as that for the estimated $\bQ_i$'s.

\section{Simulation}
\label{sec:sim}

In this section, we study the numerical performance of the proposed
matrix-valued approach. In all simulations, the observed data
$\bX_t$'s are simulated according to model \eqref{eqn:2dmodel},
\begin{equation*}
\bX_t = \bR \bF_t \bC' + \bE_t, \quad t=1,2,\ldots,T.
\end{equation*}
We choose the dimensions of the latent factor process $\bF_t$ to be
$k_1 = 3$ and $k_2=2$. The entries of $\bF_t$ are simulated as
$k_1k_2$ independent processes with noise $N(0,1)$ where the types and
coefficients of the processes will be specified later. The entries of
$\bR$ and $\bC$ are independently sampled from the uniform
distribution $U(-p_i^{-\delta_i/2},p_i^{-\delta_i/2})$ for $i=1,2$,
respectively. The error process $\bE_t$ is a white noise process with
mean $\mathbf{0}$ and a Kronecker product covariance structure, that
is, $\Cov(\rmvec(\bE_t)) = \bGamma_2 \otimes \bGamma_1$, where
$\bGamma_1$ and $\bGamma_2$ are of sizes $p_1 \times p_1$ and $p_2
\times p_2$, respectively. Both $\bGamma_1$ and $\bGamma_2$ have
values $1$ on the diagonal entries and $0.2$ on the off-diagonal
entries. For all simulations, the reported results are based on $200$
simulation runs.

We first study the performance of our proposed approach on estimating
the loading spaces. In this part, the $k_1k_2=6$ latent factors are
independent AR(1) processes with the AR coefficients
$[-0.5~0.6;~0.8~-0.4;~0.7~0.3]$. We consider three pairs of
$(\delta_1,\delta_2)$ combinations: $(0.5,0.5)$, $(0.5,0)$ and
$(0,0)$. For each pair of $\delta_1$ and $\delta_2$, the two
dimensions $(p_1,p_2)$ are chosen to be $(20,20)$, $(20,50)$ and
$(50,50)$. The sample size $T$ is selected as $0.5p_1p_2$, $p_1p_2$,
and $2p_1p_2$. We take $h_0=1$ since it is sufficient for AR(1) model
as will be shown later.

Table \ref{table:simmatdist} shows the results for estimating the
loading spaces ${\cal M}(\bQ_1)$ and ${\cal M}(\bQ_2)$. The accuracies are measured by
$\cD(\widehat{\bQ}_1,\bQ_1)$ and $\cD(\widehat{\bQ}_2,\bQ_2)$ using
the correct $k_1$ and $k_2$, respectively. The results show that with
stronger signals and more data sample points, the approach increases the
estimation accuracies. Moreover, increasing the strength of one
loading matrix can improve the estimation accuracies for both loading
spaces.

With the same simulated data, we compare the proposed matrix-valued
approach and the vector-valued approach in \citet{lam2012factor}
through the estimation accuracy of the total loading matrix $\bQ =
\bQtwo \otimes \bQone$. In what follows, the subscripts mat and vec
denote our approach and \citet{lam2012factor}'s method,
respectively. The loading space $\widehat{\bQ}_{\rm mat}$ is computed
as $\widehat{\bQ}_{\rm mat} = \widehat{\bQ}_2 \otimes \widehat{\bQ}_1$
once we obtain estimates of $\widehat{\bQ}_1$ and $\widehat{\bQ}_2$
through our approach. For the vector-valued approach, we apply
\citet{lam2012factor}'s method to the observations $\{\rmvec(\bX_t),
t=1,2,\ldots,T\}$ to obtain $\widehat{\bQ}_{\rm vec}$. Table
\ref{table:simcompdist} presents the results for the estimation
accuracies of $\bQ$ measured by $\cD_{\rm vec}(\widehat{\bQ},\bQ)$ and
$\cD_{\rm mat}(\widehat{\bQ},\bQ)$. It shows that the matrix approach
efficiently improves the estimation accuracy over the vector-valued
approach.

\begin{table}
\small
\begin{center}
\begin{tabular}{|cc|cc|cc|cc|cc|cc|}
\hline
&  &  & & \multicolumn{2}{c|}{$T=.5*p_1*p_2$}
& \multicolumn{2}{c|}{$T=p_1*p_2$}
& \multicolumn{2}{c|}{$T=2*p_1*p_2$}  \\ \hline
$\delta_1$ & $\delta_2$& $p_1$& $p_2$ & $\cD(\widehat{\bQ}_1,\bQ_1)$ &  $\cD(\widehat{\bQ}_2,\bQ_2)$ & $\cD(\widehat{\bQ}_1,\bQ_1)$ &  $\cD(\widehat{\bQ}_2,\bQ_2)$ & $\cD(\widehat{\bQ}_1,\bQ_1)$ &  $\cD(\widehat{\bQ}_2,\bQ_2)$  \\ \hline
 0.5 & 0.5 &  20 & 20&5.96(0.19)&7.12(0.03)&5.80(0.07)&7.09(0.01)&5.73(0.04)&7.08(0.01) \\
 & &  20 & 50&5.87(0.15)&7.07(0.02)&5.77(0.04)&7.05(0.01)&5.74(0.02)&7.04(0.01) \\
 & &  50 & 50&6.26(0.56)&7.05(0.01)&5.73(0.13)&7.04(0.00)&5.61(0.03)&7.03(0.00) \\ \hline
 0.5 & 0 &  20 & 20&5.36(0.41)&5.42(2.22)&4.27(1.13)&1.66(1.70)&1.52(0.75)&0.54(0.17) \\
 & &  20 & 50&5.02(0.67)&5.15(1.61)&1.82(0.77)&1.32(0.60)&0.54(0.18)&0.53(0.17) \\
 & &  50 & 50&3.68(0.48)&3.44(1.23)&1.31(0.20)&0.65(0.19)&0.51(0.07)&0.28(0.08) \\ \hline
 0 & 0 &  20 & 20&0.55(0.16)&0.44(0.10)&0.36(0.08)&0.31(0.06)&0.24(0.04)&0.22(0.04) \\
 & &  20 & 50&0.25(0.06)&0.36(0.07)&0.16(0.03)&0.26(0.05)&0.10(0.02)&0.18(0.03) \\
 & &  50 & 50&0.13(0.02)&0.12(0.02)&0.09(0.01)&0.08(0.01)&0.06(0.01)&0.06(0.01) \\ \hline
\end{tabular}
\end{center}
\caption{Means and standard deviations (in parentheses) of
  $\cD(\widehat{\bQ}_i,\bQ_i)$, $i=1,2$, over 200 simulation runs. For
  ease of presentation, all numbers in this table are the true numbers
  multiplied by 10.}
\label{table:simmatdist}
\end{table}

\begin{table}
\small
\begin{center}
\begin{tabular}{|cc|cc|cc|cc|cc|cc|cc|}
\hline
&  &  & & \multicolumn{2}{c|}{$T=.5*p_1*p_2$} & \multicolumn{2}{c|}{$T=p_1*p_2$} & \multicolumn{2}{c|}{$T=2*p_1*p_2$} \\ \hline
$\delta_1$ &  $\delta_2$& $p_1$ & $p_2$
& $\cD_{\rm vec}(\widehat{\bQ},\bQ)$ & $\cD_{\rm mat}(\widehat{\bQ},\bQ)$
& $\cD_{\rm vec}(\widehat{\bQ},\bQ)$ & $\cD_{\rm mat}(\widehat{\bQ},\bQ)$
& $\cD_{\rm vec}(\widehat{\bQ},\bQ)$ & $\cD_{\rm mat}(\widehat{\bQ},\bQ)$\\
\hline
 0.5 & 0.5 &  20 & 20&8.75(0.17)&8.26(0.07)&8.24(0.18)&8.19(0.03)&7.62(0.17)&8.16(0.02) \\
 & &  20 & 50&8.72(0.10)&8.20(0.06)&8.40(0.09)&8.15(0.01)&7.92(0.16)&8.13(0.01) \\
 & &  50 & 50&8.51(0.14)&8.34(0.22)&7.62(0.14)&8.13(0.05)&6.81(0.06)&8.09(0.01) \\ \hline
 0.5 & 0 &  20 & 20&6.40(0.29)&7.19(1.13)&5.50(0.31)&4.66(1.45)&4.37(0.45)&1.64(0.72) \\
 & &  20 & 50&5.64(0.24)&6.75(1.13)&4.75(0.35)&2.30(0.80)&3.37(0.45)&0.78(0.20) \\
 & &  50 & 50&5.07(0.10)&4.92(0.94)&4.46(0.29)&1.47(0.23)&2.73(0.46)&0.59(0.08) \\ \hline
 0 & 0 &  20 & 20&3.64(0.23)&0.71(0.16)&2.77(0.16)&0.48(0.08)&2.07(0.13)&0.33(0.04) \\
 & &  20 & 50&2.84(0.18)&0.44(0.07)&2.13(0.10)&0.30(0.05)&1.56(0.07)&0.21(0.03) \\
 & &  50 & 50&1.85(0.10)&0.18(0.02)&1.34(0.06)&0.12(0.01)&0.97(0.04)&0.09(0.01) \\ \hline
\end{tabular}
\end{center}
\caption{Means and standard deviations (in parentheses) of
  $\cD(\widehat{\bQ},\bQ)$ over 200 replicates. For ease of
  presentation, all numbers are the true numbers multiplied by 10.}
\label{table:simcompdist}
\end{table}

We next demonstrate the performance of the matrix-valued approach on
estimating the number of factors, $k_1$ and $k_2$. The data are the
same as the data in Table \ref{table:simmatdist} with
$\delta_1=\delta_2=0$ and hence the true rank pair is $(3,2)$. Table
\ref{table:simmatfreq} shows the relative frequencies of estimated
rank pairs over $200$ simulation runs. The four pairs $(2,1)$,
$(2,2)$, $(3,1)$ and $(3,2)$ have high appearances in all of the
combinations of $p_1$, $p_2$ and $T$. The row for the true rank pair
$(3,2)$ is highlighted. It shows that the relative frequency of
correctly estimating the true rank pair improves with increasing
sample size $T$. Table \ref{table:simcompfreq} shows a comparison
between the matrix and vector-valued approaches on estimating the
total number of latent factors $k=k_1k_2$. The column with the true
rank $k=6$ is highlighted. The results show that the two approaches
have similar performance when the sample size $T$ is large. For
smaller $T$, the probability of the matrix-valued approach to select
the rank pair $(3,1)$ is high and hence the frequency of estimating
the true rank $k$ decreases.

\begin{table}
\small
\centering
\begin{tabular}{|c|ccc|ccc|ccc|}
\hline
 & \multicolumn{3}{c|}{$p_1=20$,   $p_2=20$} & \multicolumn{3}{c|}{$p_1=20$,   $p_2=50$} & \multicolumn{3}{c|}{$p_1=50$,   $p_2=50$}\\ \hline
$(\hat{k}_1,\hat{k}_2) $ & $T=.5p$ & $T=p$ & $T=2p$ & $T=.5p$ & $T=p$ & $T=2p$ & $T=.5p$ & $T=p$ & $T=2p$\\ \hline
(2,1) & 0.2 & 0.055 & 0 & 0.32 & 0.005 & 0 & 0 & 0 & 0 \\
(2,2) & 0.055 & 0.04 & 0 & 0.025 & 0.005 & 0 & 0 & 0 & 0 \\
(3,1) & 0.19 & 0.215 & 0.01 & 0.47 & 0.325 & 0.005 & 0.005 & 0 & 0 \\
\rowcolor[gray]{0.8}
(3,2) & 0.365 & 0.66 & 0.985 & 0.17 & 0.665 & 0.995 & 0.995 & 1 & 1 \\
Others&0.19&0.03&0.005&0.015&0&0&0&0&0 \\ \hline
\end{tabular}
\caption{Relative frequency of estimated rank pair
  $(\hat{k}_1,\hat{k}_2)$ over $200$ runs. The row with the true rank pair
  $(3,2)$ is highlighted. Here $p=p_1p_2$.}
\label{table:simmatfreq}
\end{table}

\begin{table}
\small
\centering
\begin{tabular}{|cc|c|cc|cc|cc|cc|>{\columncolor[gray]{0.8}}c>{\columncolor[gray]{0.8}}c|cc|}
\hline
& & & \multicolumn{2}{c|}{$\hat{k}=1$} &\multicolumn{2}{c|}{$\hat{k}=2$} &\multicolumn{2}{c|}{$\hat{k}=3$} &\multicolumn{2}{c|}{$\hat{k}=4$}  &\multicolumn{2}{c|}{\cellcolor[gray]{0.8}{$\hat{k}=6$}} &\multicolumn{2}{c|}{Others} \\ \hline
$p_1$ & $p_2$ & $T$ & vec & mat & vec & mat & vec & mat & vec & mat & vec & mat & vec & mat \\ \hline
20&20&$.5p$&0.25&0.125&0.33&0.22&0.035&0.19&0.015&0.07&0.345&0.365&0.025&0.03 \\
&&$p$&0.055&0.02&0.105&0.06&0&0.215&0&0.045&0.83&0.66&0.01&0 \\
&&$2p$&0&0&0.005&0&0&0.01&0&0.005&0.995&0.985&0&0 \\ \hline
20&50&$.5p$&0.03&0.015&0.62&0.32&0&0.47&0&0.025&0.34&0.17&0.01&0 \\
&&$p$&0&0&0.14&0.005&0&0.325&0&0.005&0.86&0.665&0&0 \\
&&$2p$&0&0&0&0&0&0.005&0&0&1&0.995&0&0 \\ \hline
50&50&$.5p$&0.07&0&0&0&0&0.005&0&0&0.93&0.995&0&0 \\
&&$p$&0&0&0&0&0&0&0&0&1&1&0&0 \\
&&$2p$&0&0&0&0&0&0&0&0&1&1&0&0 \\ \hline
\end{tabular}
\caption{Relative frequency of estimated total rank $\hat{k}$ over
  $200$ replicates for both the vector and matrix-valued
  approaches. The column with the true rank $k=6$ is highlighted. Here
  $p=p_1p_2$.}
\label{table:simcompfreq}
\end{table}

We now study the effects of the lag parameter $h_0$. The $k_1k_2$
factors are assumed to be
independent and follow the same model which is either an
AR(1) or an MA(2) model. For the AR(1) model, the coefficients of all
the factors are $0.9$, $0.6$ or $0.3$. For the MA(2) model, we
consider the case $f_t = e_t + 0.9 e_{t-2}$. We take $\delta_1 =
\delta_2 = 0$, $T=p_1p_2$ and compare the estimation accuracies of the
two loading spaces for four lag choices, $h_0=1,2,3,4$. Table
\ref{table:hzero} shows the results of $\cD(\widehat{\bQ}_1,\bQ_1)$
and $\cD(\widehat{\bQ}_2,\bQ_2)$. It is seen that, for the AR(1)
processes, taking
$h_0=1$ is sufficient. Larger $h_0$ in fact decreases the performance,
especially for small AR coefficient cases. Note that
larger $h_0$ increases the
signal strength in the matrix $\bM$, but also increases the noise level in
its sample version $\hat{\bM}$.
For an AR(1) model with small AR coefficient, the autocorrelation in higher
lags is relatively small hence the additional signal strength is limited.
For the MA(2) process, one must use $h_0\geq 2$ since lag 1 autocovariance
 matrix is zero and does not provide any information.
$h_0=2$ performs the best, since  all higher lags carry no additional
information, but add significant amount of noise.

\begin{table}
\tiny
\centering
\begin{tabular}{|c|cc|cc|cc|cc|cc|}
\hline
& & & \multicolumn{2}{c|}{$h_0=1$} & \multicolumn{2}{c|}{$h_0=2$} & \multicolumn{2}{c|}{$h_0=3$}  & \multicolumn{2}{c|}{$h_0=4$} \\ \hline
AR(1) & $p_1$& $p_2$ & $\cD(\widehat{\bQ}_1,\bQ_1)$ &  $\cD(\widehat{\bQ}_2,\bQ_2)$ & $\cD(\widehat{\bQ}_1,\bQ_1)$ &  $\cD(\widehat{\bQ}_2,\bQ_2)$ & $\cD(\widehat{\bQ}_1,\bQ_1)$ &  $\cD(\widehat{\bQ}_2,\bQ_2)$ & $\cD(\widehat{\bQ}_1,\bQ_1)$ &  $\cD(\widehat{\bQ}_2,\bQ_2)$ \\ \hline
 0.9 &  20 & 20&0.13(0.02)&0.10(0.02)&0.13(0.02)&0.10(0.02)&0.14(0.02)&0.10(0.02)&0.14(0.02)&0.11(0.02) \\
 &  20 & 50&0.05(0.01)&0.07(0.01)&0.05(0.01)&0.07(0.01)&0.05(0.01)&0.07(0.01)&0.06(0.01)&0.08(0.01) \\
 &  50 & 50&0.03(0.00)&0.03(0.00)&0.03(0.00)&0.03(0.00)&0.03(0.00)&0.03(0.00)&0.03(0.00)&0.03(0.00) \\ \hline
 0.6 &  20 & 20&0.36(0.07)&0.26(0.04)&0.41(0.08)&0.27(0.05)&0.47(0.10)&0.28(0.05)&0.53(0.12)&0.29(0.05) \\
 &  20 & 50&0.15(0.03)&0.19(0.02)&0.18(0.04)&0.20(0.03)&0.21(0.05)&0.21(0.03)&0.24(0.06)&0.21(0.03) \\
 &  50 & 50&0.09(0.01)&0.07(0.01)&0.10(0.01)&0.08(0.01)&0.11(0.02)&0.08(0.01)&0.12(0.02)&0.08(0.01) \\ \hline
 0.3 &  20 & 20&1.56(0.72)&0.57(0.13)&2.31(0.99)&0.60(0.16)&2.82(1.04)&0.63(0.18)&3.12(1.05)&0.67(0.18) \\
 &  20 & 50&0.64(0.21)&0.45(0.10)&1.14(0.49)&0.49(0.13)&1.66(0.68)&0.55(0.19)&2.13(0.82)&0.63(0.24) \\
 &  50 & 50&0.26(0.05)&0.17(0.03)&0.39(0.08)&0.18(0.03)&0.56(0.11)&0.20(0.04)&0.74(0.14)&0.21(0.04) \\ \hline
 MA(2) &  20 & 20&2.60(1.11)&0.88(0.28)&0.48(0.12)&0.27(0.05)&0.59(0.15)&0.28(0.05)&0.68(0.17)&0.28(0.06) \\
 &  20 & 50&2.76(1.16)&1.13(0.56)&0.21(0.04)&0.21(0.03)&0.27(0.06)&0.22(0.04)&0.32(0.07)&0.22(0.04) \\
 &  50 & 50&2.85(1.15)&0.68(0.23)&0.11(0.02)&0.08(0.01)&0.13(0.02)&0.08(0.01)&0.15(0.02)&0.08(0.01) \\ \hline
\end{tabular}
\caption{Means and standard deviations (in parentheses) of
  $\cD(\widehat{\bQ}_1,\bQ_1)$ and $\cD(\widehat{\bQ}_2,\bQ_2)$ for
  different lag parameter $h_0$. All numbers are the true numbers
  multiplied by 10.}
\label{table:hzero}
\end{table}

 Next  we study
the performance of recovering the signal $\bS_t$. The latent factors
are simulated in the same way as the data in Table
\ref{table:simmatdist}. We take $\delta_1=\delta_2=0$,
$p_1=p_2=10,20,50$, and $T=50,200,1000,5000$. The recovery accuracy of
$\widehat{\bS}_t$, denoted by $\cD(\widehat{\bS},\bS)$, is estimated
by the average of $\| \widehat{\bS}_t - \bS_t\|_2$ for
$t=1,2,\ldots,T$ further normalized by $\sqrt{p_1p_2}$, that is,
$\cD(\widehat{\bS},\bS) = p_1^{-1/2}p_2^{-1/2}\big(\sum_{t=1}^T \|
\widehat{\bS}_t -\bS_t\|_2 / T \big)$.  Table \ref{table:simshat}
presents the results of $\cD(\widehat{\bS},\bS)$ and
$\cD(\widehat{\bQ},\bQ)$ for the two approaches. It shows that, when $T$
is relatively large (hence the $\bQ$ is estimated relatively accurately),
increasing $p$ improves the estimation of $\bS$. For the same $p$ and
relatively large $T$, further increasing
$T$ has a limited benefit in improving the estimation of $\bS$. The
estimation accuracies of $\bS$ of the proposed matrix-valued
approach are better than that of the vector-valued approach, though the
relative improvement decreases as $T$ increases, even the improvement of
estimating $\bQ$ is significant.

\begin{table}
\centering
\begin{tabular}{|c|c|cc|cc|}
\hline
$p_1 = p_2$ & $T$ & $\cD_{\rm vec}(\widehat{\bQ},\bQ)$ & $\cD_{\rm mat}(\widehat{\bQ},\bQ)$ & $\cD_{\rm vec}(\widehat{\bS},\bS)$ & $\cD_{\rm mat}(\widehat{\bS},\bS)$ \\ \hline
 10 &  50&6.26(0.38)&3.66(0.92)&4.05(0.28)&3.41(0.39) \\
 &  200&4.11(0.33)&1.42(0.42)&3.02(0.19)&2.62(0.15) \\
 &  1000&2.12(0.13)&0.50(0.09)&2.48(0.05)&2.40(0.04) \\
 &  5000&0.99(0.05)&0.21(0.03)&2.38(0.02)&2.36(0.01) \\ \hline
 20 &  50&5.65(0.39)&2.36(1.26)&3.07(0.27)&1.86(0.61) \\
 &  200&3.64(0.23)&0.71(0.16)&1.96(0.16)&1.11(0.05) \\
 &  1000&1.88(0.11)&0.29(0.04)&1.25(0.05)&1.02(0.01) \\
 &  5000&0.87(0.03)&0.13(0.01)&1.04(0.01)&1.00(0.00) \\ \hline
 50 &  50&5.81(0.35)&3.17(1.47)&2.64(0.26)&1.95(0.79) \\
 &  200&3.78(0.24)&0.62(0.19)&1.57(0.16)&0.56(0.09) \\
 &  1000&2.04(0.10)&0.21(0.03)&0.82(0.06)&0.42(0.01) \\
 &  5000&0.97(0.04)&0.09(0.01)&0.49(0.02)&0.40(0.00) \\ \hline
\end{tabular}
\caption{Means and standard deviations (in parentheses) of estimation
  accuracies of $\bQ$ and $\bS$. All numbers are the true numbers
  multiplied by 10.}
\label{table:simshat}
\end{table}

Next, we conduct a $10$-fold cross-validation study. The data are
generated in the same way as the data in Table \ref{table:simmatdist}
with $\delta_1=\delta_2=0$, $p_1=p_2=20$ and $T=1000$. We vary the
estimated number of factors $k_1$ and $k_2$ from all combinations of
$k_1=1,2,3,4$ and $k_2=1,2,3$. The means of the
out-of-sample RSS/SST are reported in Table \ref{table:simcv}. For the
matrix-valued approach, the RSS/SST decreases rapidly when
$k_1$ and $k_2$ increase, before they reach the true rank pair
$(3,2)$ (highlighted in the table).
Then the RSS/SST value remain roughly the same with increasing
estimated ranks when $k_1 > 3$ and $k_2 > 2$. For the
vector-valued approach, $k=k_1k_2$, hence the values in the table
are the same for the same $k_1k_2$ value
(e.g. $(2,3)$ and $(3,2)$ are equivalent).
Its performance improves quickly
as $k$ increases until $k=6$, the true number of
factors. Then the performance remains relatively the same for
$\hat{k} > 6$.

\begin{table}
\centering
\begin{tabular}{|c|cc|cc|cc|}
\hline
& \multicolumn{2}{c|}{$\hat{k}_2=1$} & \multicolumn{2}{c|}{$\hat{k}_2=2$} & \multicolumn{2}{c|}{$\hat{k}_2=3$} \\ \hline
$\hat{k}_1$ & vec & mat & vec & mat & vec & mat \\ \hline
 1&0.83&0.83&0.72&0.75&0.63&0.75 \\ \hline
 2&0.72&0.71&0.58&0.56&0.49&0.55 \\ \hline
 3&0.63&0.67& {\cellcolor[gray]{0.8} 0.49}& {\cellcolor[gray]{0.8} 0.47}&0.46&0.46 \\ \hline
 4&0.58&0.66&0.46&0.47&0.44&0.43 \\ \hline
\end{tabular}
\caption{Means of out-of-sample RSS/SST for $10$-fold cross-validation
  over $200$ simulation runs. The cell corresponding to the true order
$(3,2)$ is highlighted.}
\label{table:simcv}
\end{table}

\section{Real Example: Fama-French 10 by 10 Series}

In this section we illustrate the matrix factor model using the
Fama-French 10 by 10 return series. A universe of
stocks is grouped into 100 portfolios, according to
ten  levels of market capital (size) and ten
levels of book to equity ratio (BE). Their monthly returns
from January 1964 to December, 2015 for total 624 months
and overall 62,400 observations are used in this analysis.
For more detailed information, see
{\it http://mba.tuck.dartmouth.edu/pages/faculty/ken.french/data\_library.html}.

All the 100 series are clearly related to the overall market condition.
In this analysis we simply subtract the corresponding monthly excess
market return from each of the series, resulting in 100 market-adjusted
return series. We chose not to fit a standard
CAPM model to each of the series to remove the market effect, as it will
involve estimating 100 different betas.
The market return data are obtained from the same website above.

Figure \ref{FFtsplot} shows the time series plot of the 100 series
(standardized), and Figures \ref{FFeigen1} and \ref{FFeigen2}
show the logarithms and ratios of eigenvalues of
$\widehat{\bM}_1$ and $\widehat{\bM}_2$ for the row (size) and column (BE) loading
matrices. Since the series shows very small autocorrelation beyond $h=1$,
in this example we use $h_0=1$.  The results by using $h_0=2$ are similar.
Although the eigenvalue ratio estimate presented in Section
\ref{subsec:est} indicates $k_1=k_2=1$, we use $k_1=k_2=2$ here for
illustration. Tables \ref{FFsize} and \ref{FFbe} show the estimated
loading matrices after a varimax rotation that maximizes the variance
of the squared factor loadings,
scaled by 30 for a cleaner view.
For size, it is seen that there are possibly two or three groups.
The $1$-st to $5$-th smallest size
portfolios load heavily (with roughly equal weights) on the first row
of the factor matrix, while the $6$-th to $9$-th smallest size
portfolios load heavily (with roughly equal weights) on the second row
of the factor matrix. The largest ($10$-th) size portfolio behaves similar
to the other larger size portfolios, but with some differences.
We note
that the Fama-French size factor proposed in \citet{Fama&French93} is
constructed using the return differences of the largest 30\% of the
companies (combining the $8$-th to $10$-th size portfolio) and the
smallest 30\% of the companies (combining our $1$st to $3$rd size
portfolio).

Turning to the book to equity ratio, Table \ref{FFbe} shows a
different pattern in the column loading matrix.
There seem to have three groups. The smallest $2$-nd to
$4$-th BE portfolios load heavily on the first
column of the factor matrix; the $5$th to $10$th BE portfolios load
heavily on the second columns of the factor matrix. The smaller
($1$st) BE portfolios load heavily on both columns of the factor matrix,
with different loading coefficients.

Figure \ref{FFfactor} shows the estimated factor matrices over time. It
can be potentially used to replace the Fama-French size factor (SMB) and
book to equity factor (HML) in a Fama-French factor model for asset
pricing, factor trading and other usage, though further analysis is needed
to assess their effectiveness.
Cross-correlation study shows that there are not many significant
cross-correlation of lag larger than 0 among the factors, though the
factors show some strong contemporary correlation as the factor matrices
are subject to rotation -- in our case we performed rotation to reveal the
group structure in the loading matrices. A principle component analysis of the
four factor series reveals that three principle components can explain
98\% of the variation in the four factors, hence there may still be
some redundancy in the factors and the model may be further simplified.

\begin{table}
\begin{center}
\begin{tabular}{c|cccccccccc}
Factor & S1 & S2 & S3 & S4 & S5 & S6 & S7 & S8 & S9 & S10 \\ \hline
1 & \cellcolor[gray]{0.8}{-13} & \cellcolor[gray]{0.8}{-14} &
\cellcolor[gray]{0.8}{-13} & \cellcolor[gray]{0.8}{-13} &
\cellcolor[gray]{0.8}{-10}
& -5 & -2 & 1& 6&  \cellcolor[gray]{0.6}{7} \\
2 & 0 & 0 & -2 & 3 & 5 & \cellcolor[gray]{0.8}{12} &
\cellcolor[gray]{0.8}{12} & \cellcolor[gray]{0.8}{18} &
\cellcolor[gray]{0.8}{15}  & \cellcolor[gray]{0.6}{5}
\end{tabular}
\end{center}
\caption{Fama-French series: Size loading matrix after rotation and scaling.}
\label{FFsize}
\end{table}

\begin{table}
\begin{center}
\begin{tabular}{c|cccccccccc}
Factor & BE1 & BE2 & BE3 & BE4 & BE5 & BE6 & BE7 & BE8 & BE9 & BE10 \\ \hline
1 & {\cellcolor[gray]{0.6} -21}  & {\cellcolor[gray]{0.8} -14}
& {\cellcolor[gray]{0.8} -11} & {\cellcolor[gray]{0.8} -9}
& -4 & -1 & -1 & -4 & 1 & 3 \\
2 & {\cellcolor[gray]{0.6} -9} &  2 & 3  & 7 & {\cellcolor[gray]{0.8} 9}
& {\cellcolor[gray]{0.8} 10} & {\cellcolor[gray]{0.8} 10}
& {\cellcolor[gray]{0.8} 10} & {\cellcolor[gray]{0.8} 13}
& {\cellcolor[gray]{0.8} 14}
\end{tabular}
\end{center}
\caption{Fama-French series: BtoE loading matrix after rotation and scaling.}
\label{FFbe}
\end{table}

\begin{figure}
\centerline{\includegraphics[width=6.3in,height=5.0in]{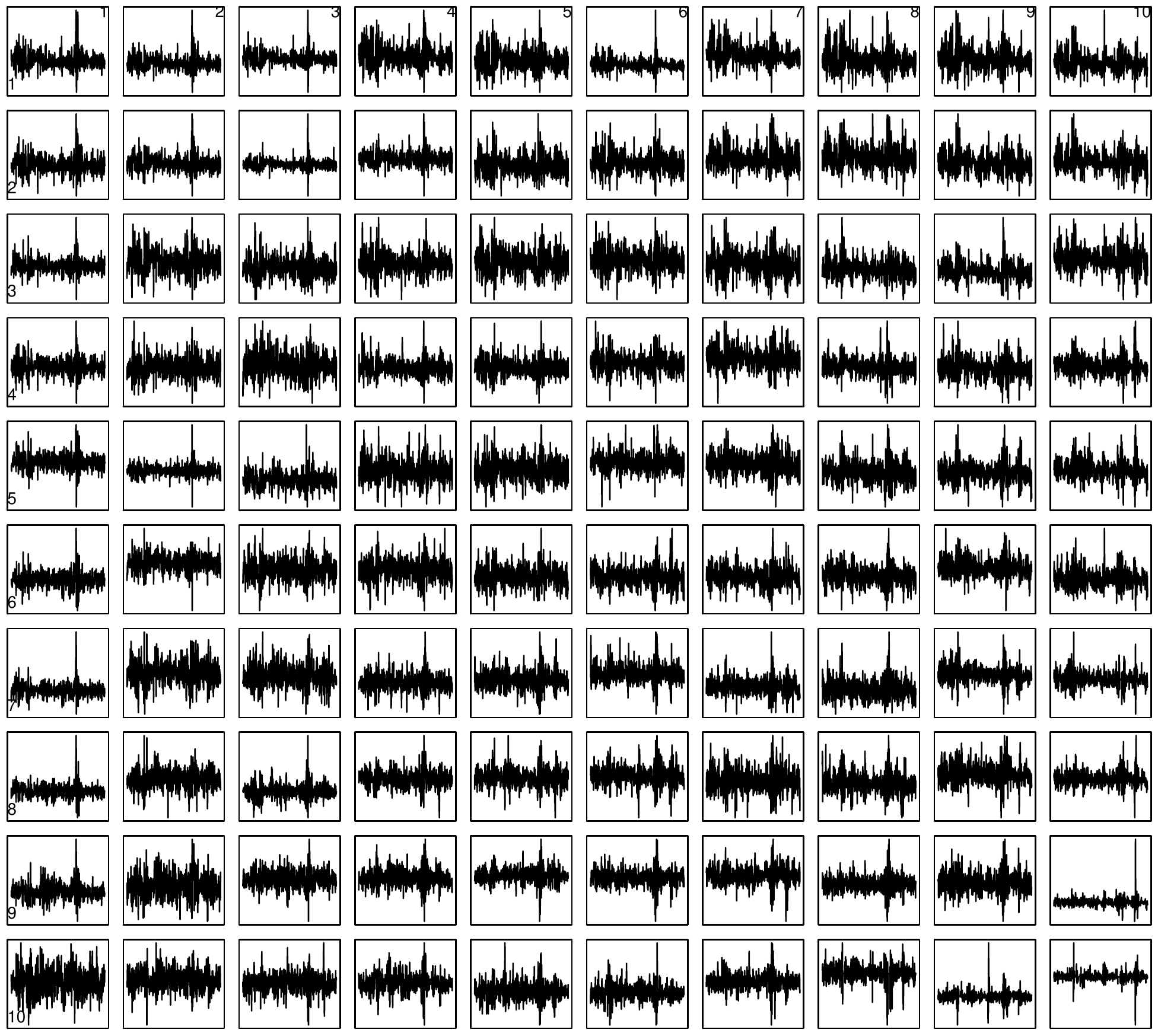}}
\caption{Time series plot of Fama-French 10 by 10 series.}
 \label{FFtsplot}
\end{figure}

\begin{figure}
\centerline{\includegraphics[width=3.2in,height=3.2in]{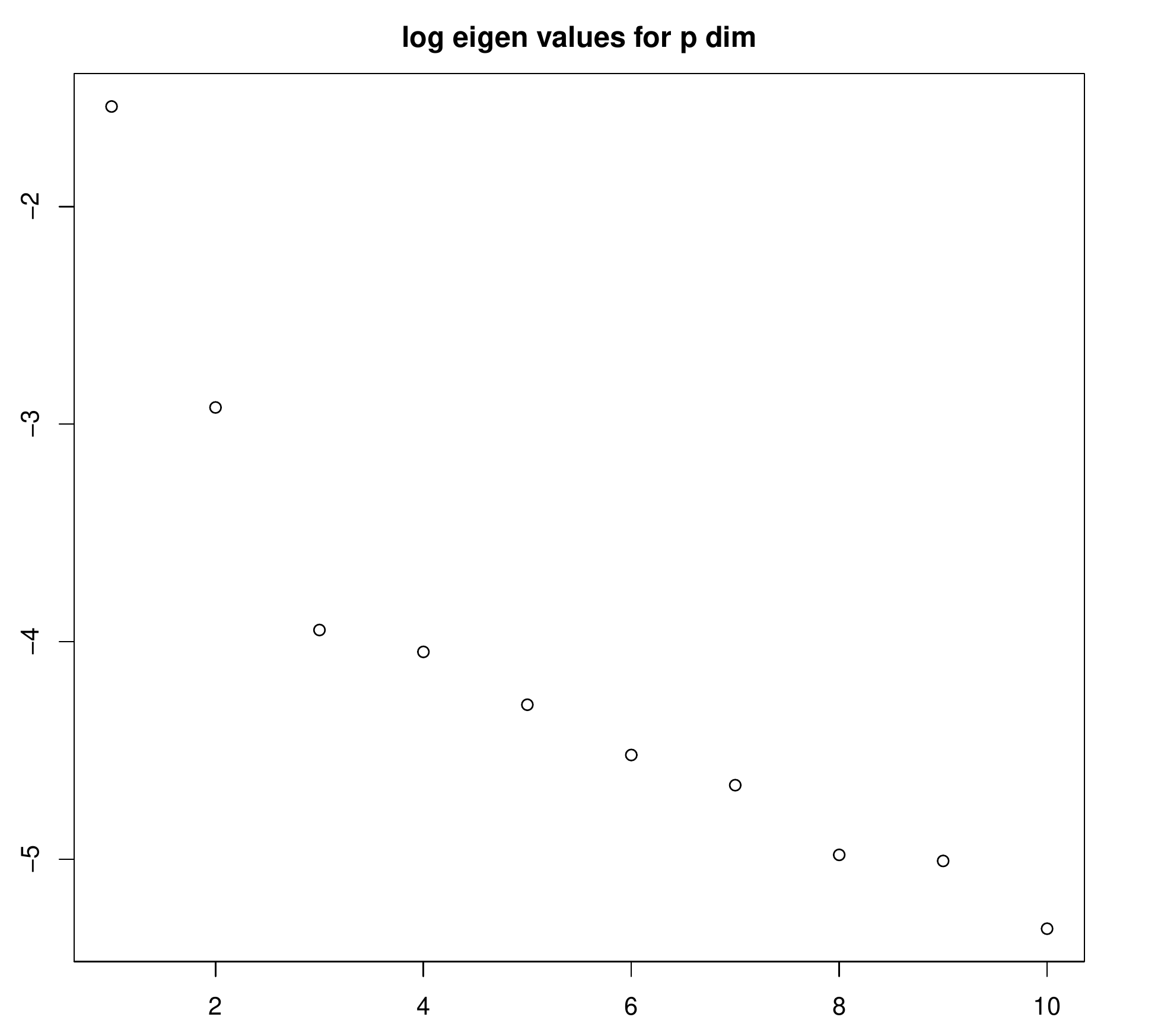}
\includegraphics[width=3.2in,height=3.3in]{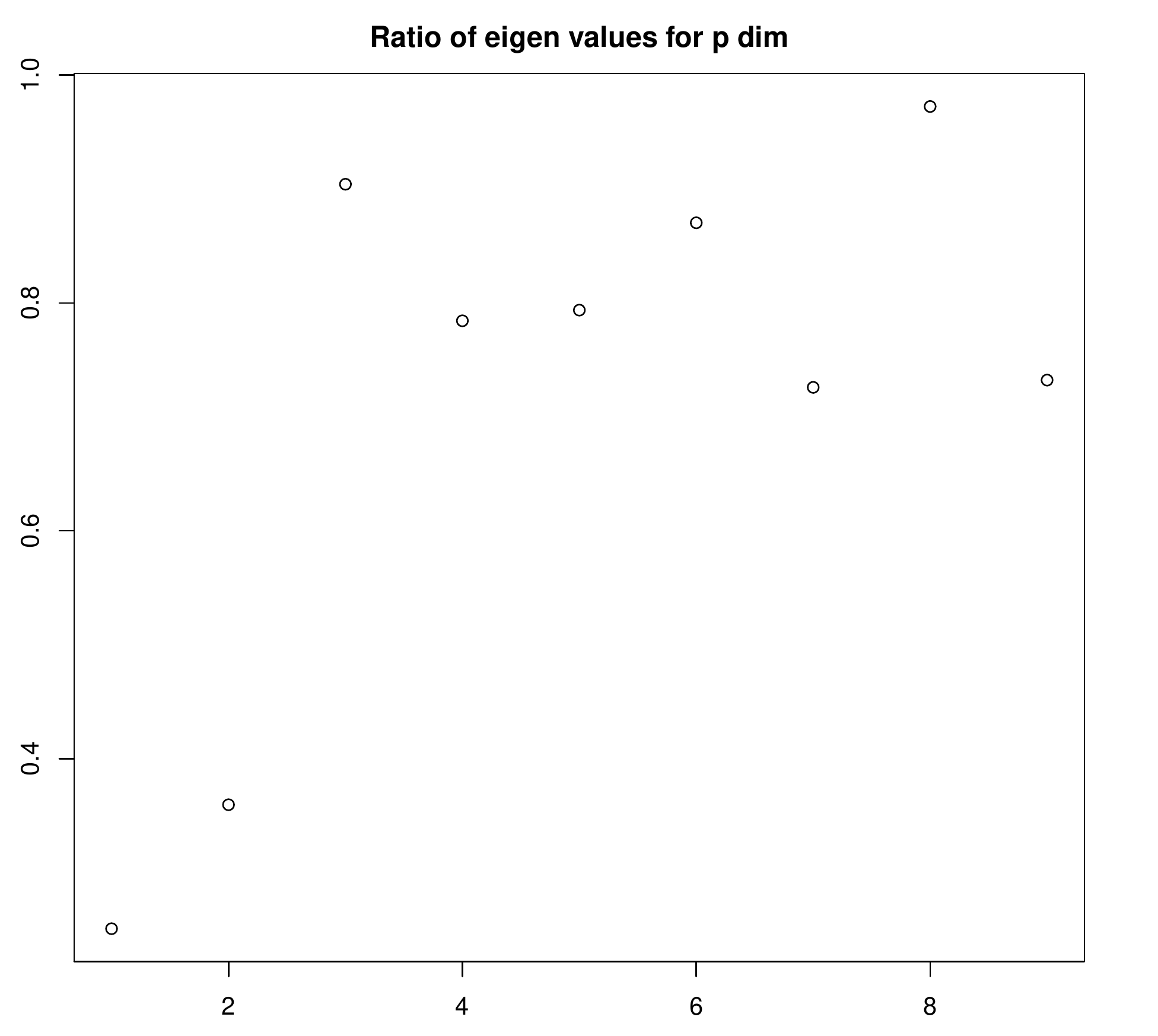}}
\caption{Fama-French series: Logarithms and ratios of eigenvalues of
$\widehat{\bM}_1$ for the row (size) loading  matrix.}
\label{FFeigen1}
\end{figure}

\begin{figure}
\centerline{\includegraphics[width=3.2in,height=3.2in]{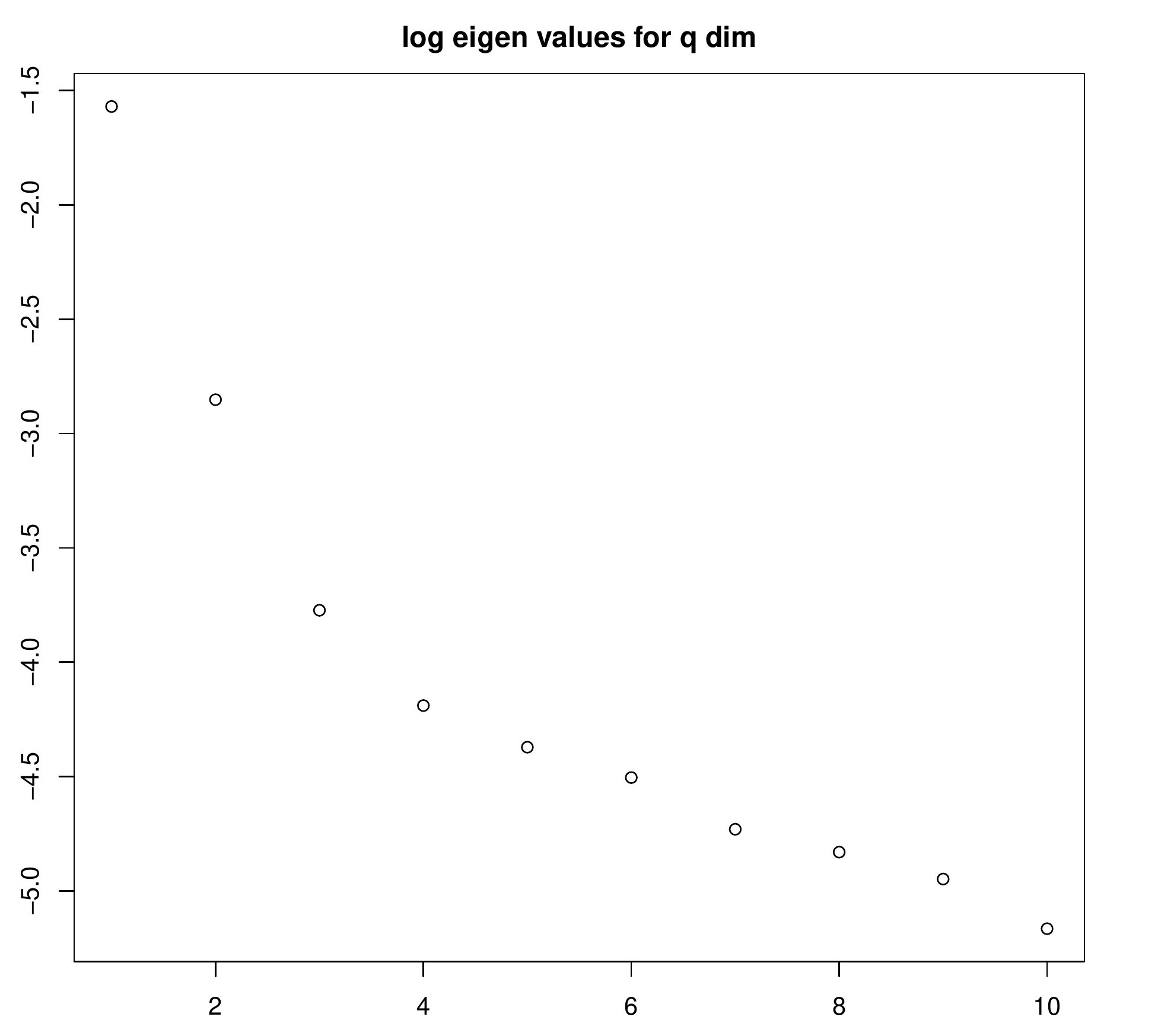}
\includegraphics[width=3.2in,height=3.2in]{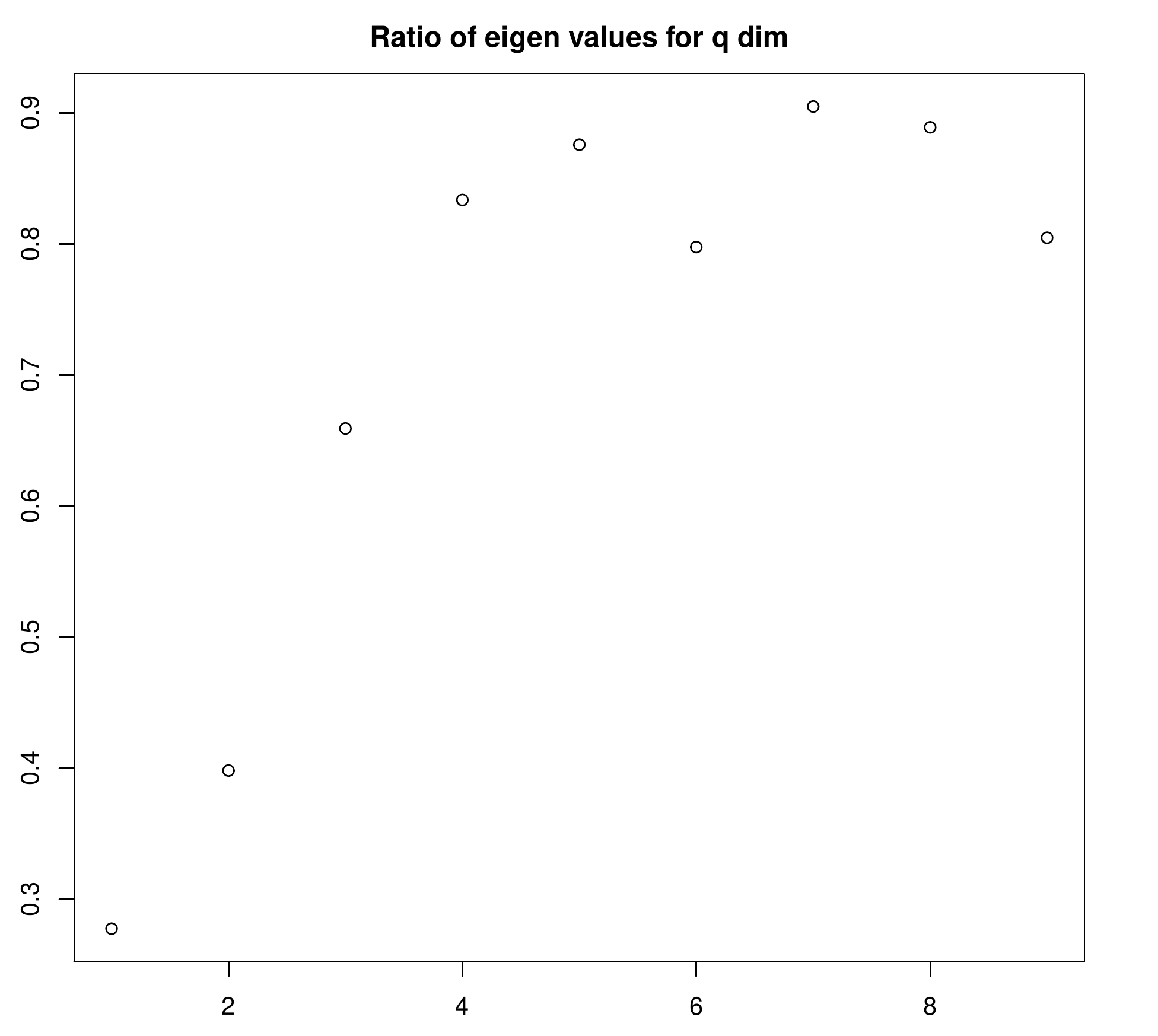}}
\caption{Fama-French series: Logarithms and ratios of eigenvalues of
$\widehat{\bM}_2$ for the column (BE) loading matrix.}
\label{FFeigen2}
\end{figure}


\begin{figure}
\centerline{\includegraphics[width=6.0in,height=5.0in]{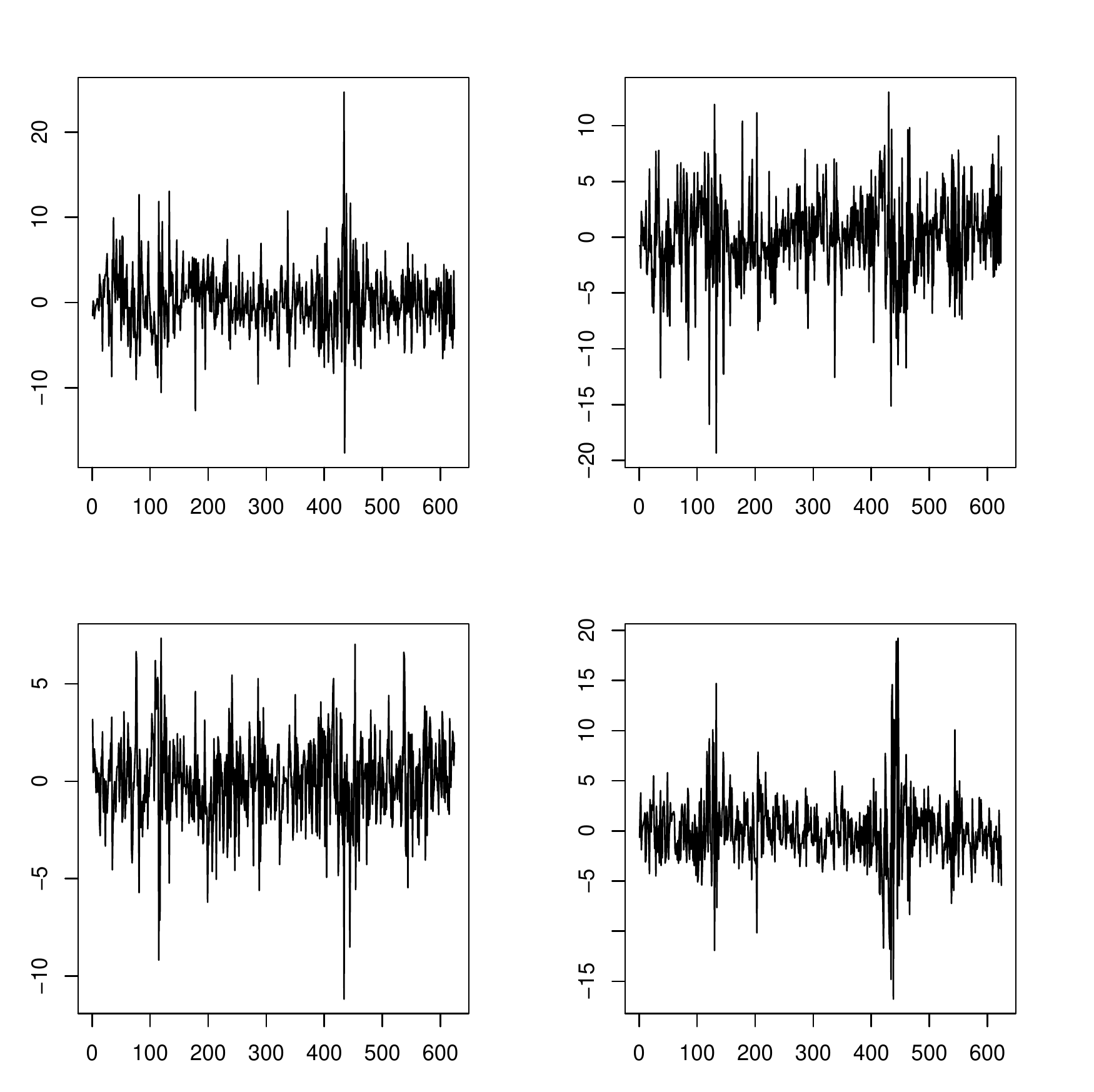}}
\caption{Fama-French series: Estimated factors.}
\label{FFfactor}
\end{figure}

\begin{figure}
\centerline{\includegraphics[width=3.2in,height=3.2in]{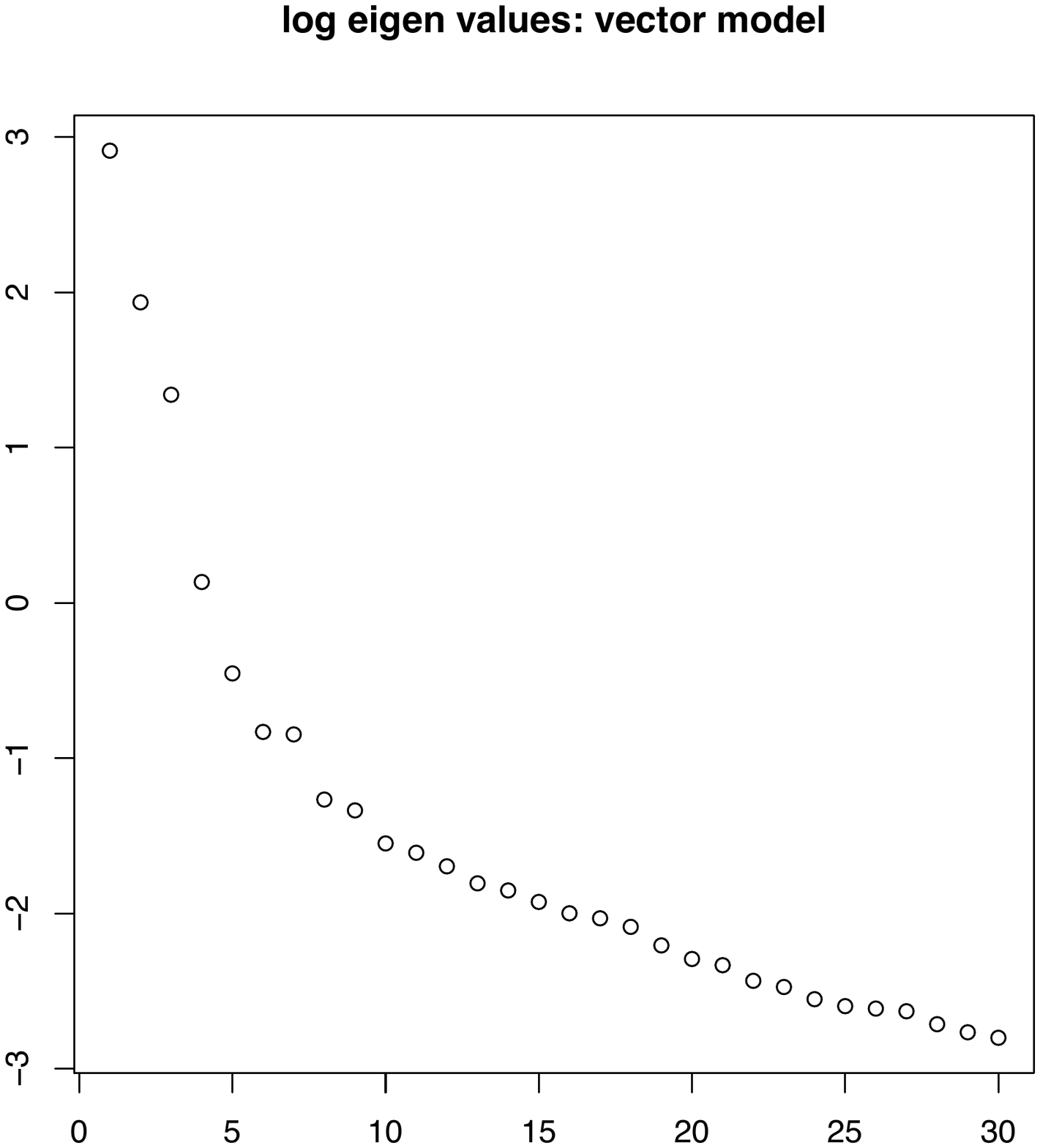}
\includegraphics[width=3.2in,height=3.2in]{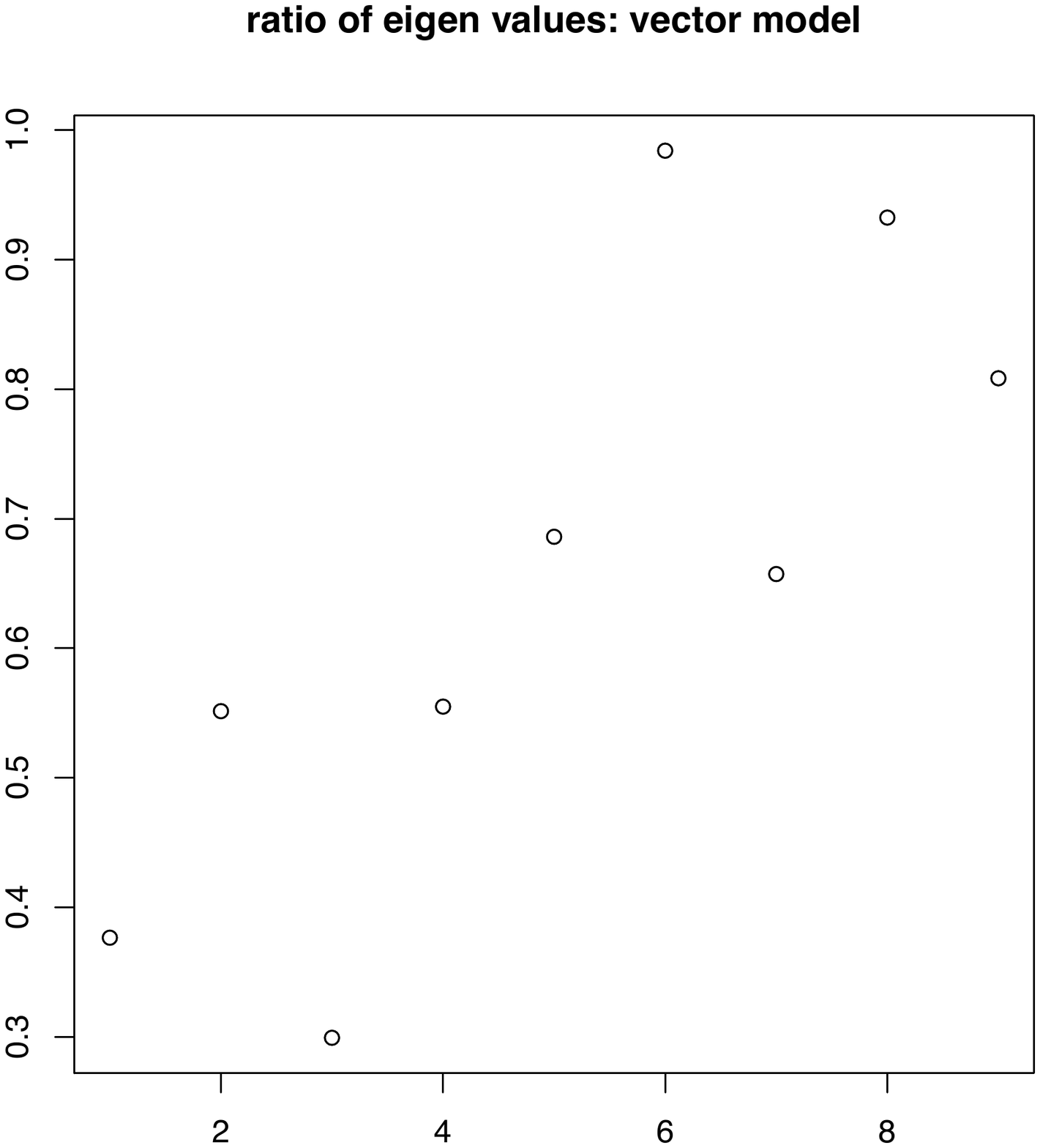}}
\caption{Fama-French series: Logarithms and ratios of eigenvalues of $\bM$ for the vectorized factor model.}
\label{FFeign3}
\end{figure}

\begin{table}
\begin{center}
\begin{tabular}{c|ccccc}
 & factor & RSS & \# factors & \# parameters \\ \hline
Matrix model & (0,0) & 29,193 & 0 & 0 \\
Matrix model & (2,2) & 14,973 & 4 & 40 \\
Matrix model & (2,3) & 14,514 & 6 & 50 \\
Matrix model & (3,2) & 14,166 & 6 & 50 \\
Matrix model & (3,3) & 13,530 & 9 & 60 \\ \hline
Vector model &    3  & 16,262 & 3 & 300 \\
Vector model &    4  & 15,365 & 4 & 400 \\
Vector model &    5  & 14,565 & 5 & 500 \\
Vector model &    6  & 14,149 & 6 & 600 \\
\end{tabular}
\end{center}
\caption{Comparison of different models for Fama-French series.}
\label{table:FFcompare}
\end{table}

Figure \ref{FFeign3} shows the logarithms and ratios of eigenvalues of
$\widehat{\bM}$ in \cite{lam2011estimation} for a vectorized factor
model (\ref{eqn:1dmodel}).  Models with various number of factors were
estimated and a comparison is shown in Table \ref{table:FFcompare}
using a version of rolling-validation. Specifically, for each year
between 1996 to 2015, we use all data available before the year to fit
a matrix (or vector) factor model and estimate the corresponding
loading matrices. Using these estimated loading matrices and the
observed 12 months of the data in the year, we estimate the factors
and the corresponding residuals.  Total sum of squares of the 12
residuals of the 100 series of the 20 years are reported. The RSS
corresponding to model $(0,0)$ is the total sum of squares of the
observed 100 series of the 20 years being studied.  It is seen that
the matrix factor model with $(2,2)$ factor matrices performs better
than the vectorized factor model with equal number of factors and many
more parameters in the loading matrices.  The $(3,2)$ matrix factor
model performs similarly as the 6-factor vectorized factor model, but
the number of parameters used is much smaller.

\section{Real Example: Series of Company Financials}

In this example we analyze the series of financial data reported by a group of
200 companies. We constructed 16 financial characteristics based on company
quarterly financial reports. The list of variables and their definitions is
given in Appendix 2. The period is from the first quarter of 2006 to
the fourth
quarter of 2015 for 10 years with total
40 observations. The total number of time series is 3,200.

Figures \ref{CFeign1} and \ref{CFeign2}
show the eigenvalues and their ratios of $\widehat{\bM}_1$ and $\widehat{\bM}_2$ for row factors and column factors. The estimated
dimensions $k_1$ and $k_2$ are both 3, though we use $k_1=5$ and $k_2=20$
for this illustration, with interesting results. Estimation is done using
$h_0=2$.

\begin{figure}
\centerline{\includegraphics[width=6.8in,height=3.0in]{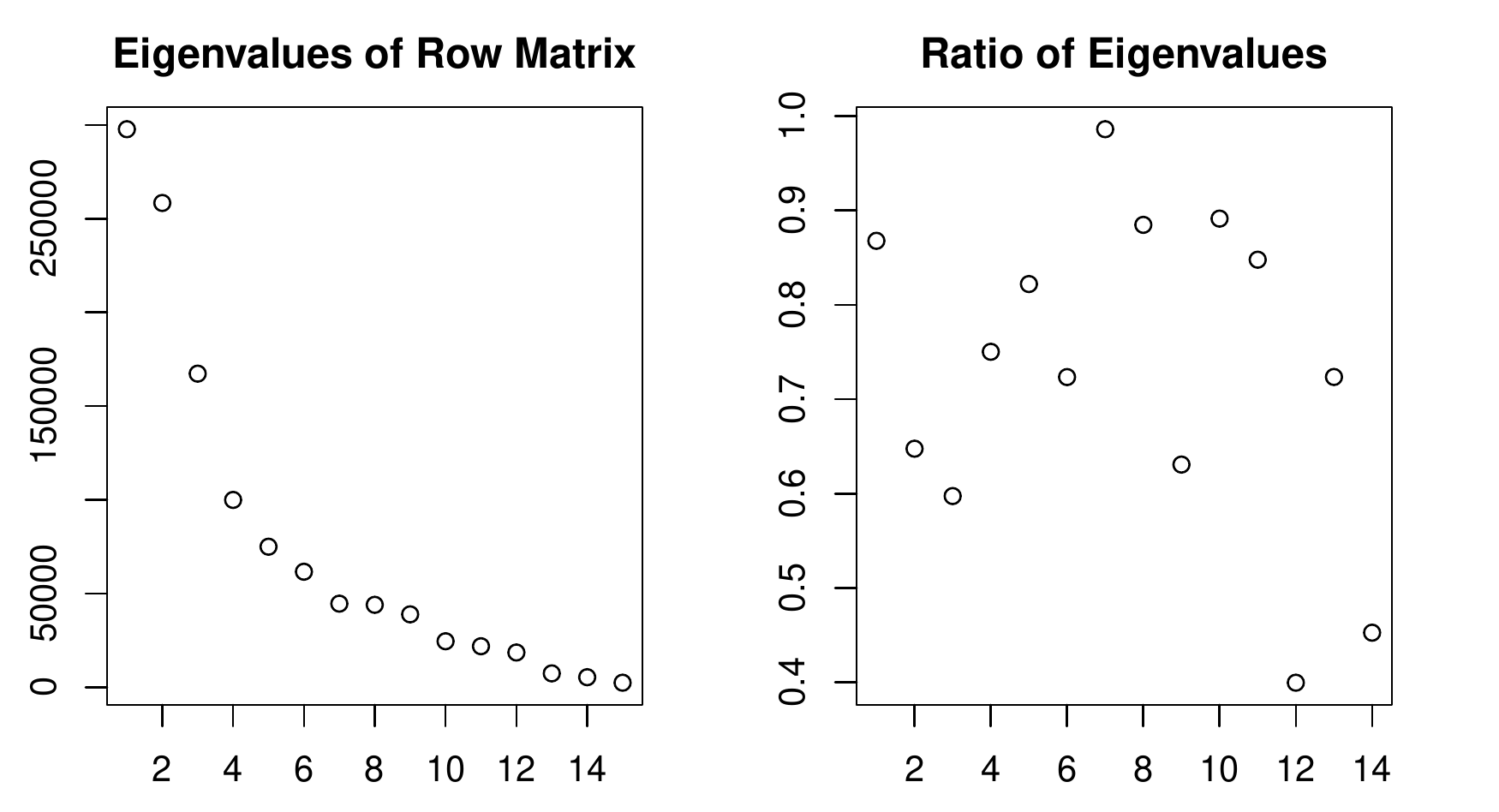}}
\caption{Financial series: Eigenvalues and their ratios
of $\widehat{\bM}_1$. }
\label{CFeign1}
\end{figure}

\begin{figure}
\centerline{\includegraphics[width=6.8in,height=3.0in]{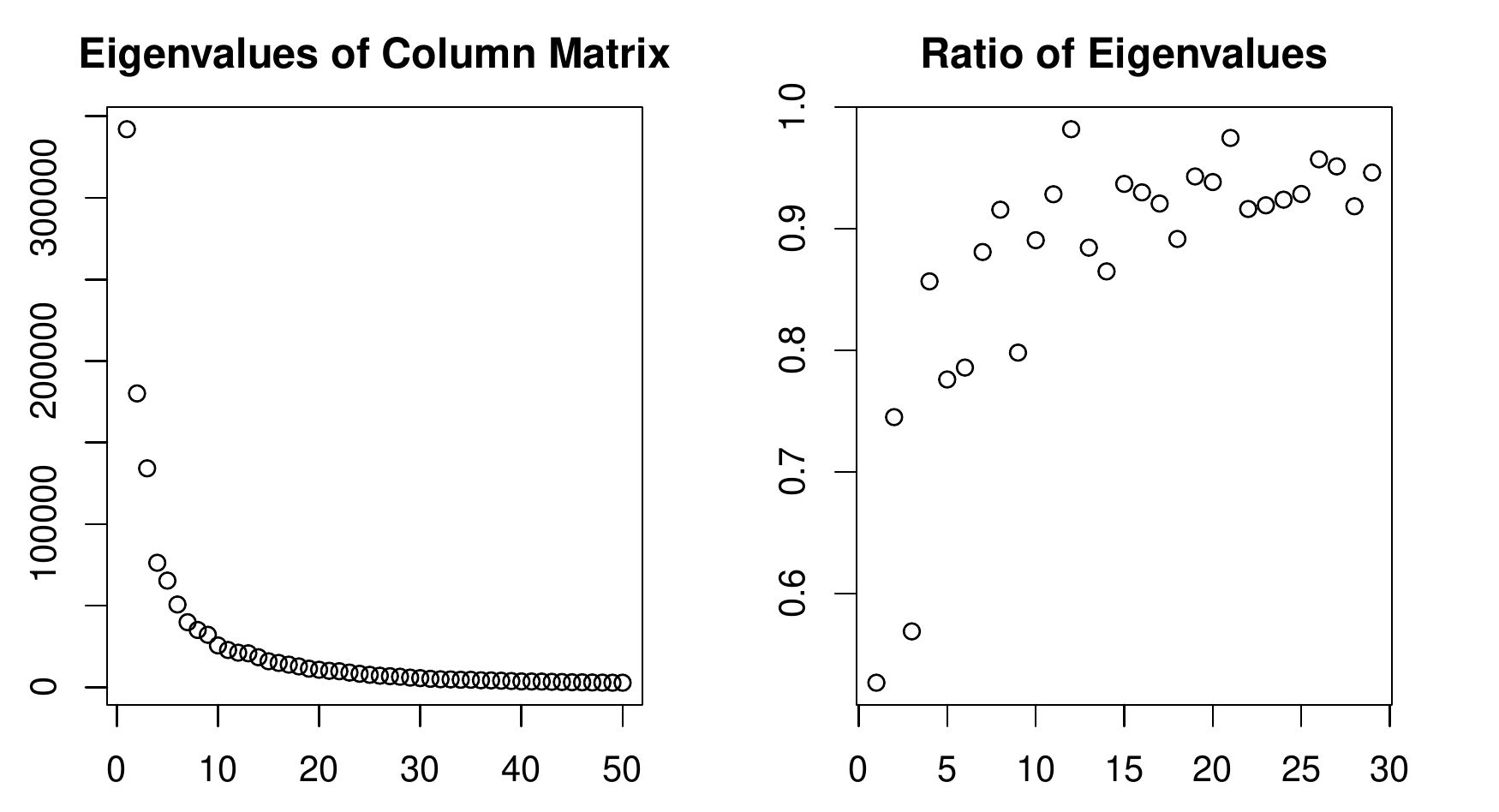}}
\caption{Financial series: Eigenvalues and their ratios of
$\widehat{\bM}_2$. }
\label{CFeign2}
\end{figure}

The estimated row loading matrix is rotated to maximize its variance, with
potential grouping shown by the shaded areas in
Table \ref{fin:group.loading}.
It shows the loading of each financial on the five rows
of the factor matrix,
after proper scaling (30 times) and reordering for easy visualization. The
two financials in Group 1 load almost exclusively on Row 1 of the factor
matrix, with almost the same weights.
The six financials in Group 2 load heavily on Row 2,
again with almost the same weights.  The three
financials in Group 3 load on Rows 3 and 4, with somewhat different weights.
Finally, the five financials in Group 4 mainly load on Row 5 of the factor
matrix, with Payout.Ratio having opposite weights from the others.

The detailed grouping is shown in Table \ref{fin:group}.
Group 1 consists of
asset to equity ratio and liability to equity ratio. They are
two very closely related measures. Group 2 consists of
six measures on earnings and returns. Group 3 consists
of cash and revenue per share, and gross margin.
Group 4 consists of profit growth and revenue growth comparing to
the previous quarter and the same quarter last year. The Payout Ratio
variable is also included in the group.
Such groupings are relatively expected.

\begin{table}
\begin{center}
\begin{tabular}{c|cccccccccccccccc}
Row Factor & F1 & F2 & F3 & F4 & F5 & F6 & F7 & F8 & F9 & F10
&F11 & F12 & F13 & F14 &F15 &F16 \\ \hline
1 &
\cellcolor[gray]{0.8}{21} & \cellcolor[gray]{0.8}{21} &	-1 & -1 & -1 & -1
& -1 & 	4 & -2 & 2 & 1 & 0 & 1 & 0 & 0 & -1 \\
2 &
1 & 1 & \cellcolor[gray]{0.8}{-13} &\cellcolor[gray]{0.8}{-9}&
\cellcolor[gray]{0.8}{-11} & \cellcolor[gray]{0.8}{-11}
&\cellcolor[gray]{0.8}{-13} & \cellcolor[gray]{0.8}{-12}
& 7 & -6 & -1 & 3 & -1 & 1 & 0 & 1 \\
3 &
0 & 0 & 0 & -10 & 0 & -1 & 0 & -1 &
\cellcolor[gray]{0.8}{-11} & \cellcolor[gray]{0.8}{9} &
\cellcolor[gray]{0.8}{-24} & -4  & -1 & 2 & -2 & 0 \\
4 &
0 & 0 & -3 & 2 & 5 & 4 & -2 & -2 &
\cellcolor[gray]{0.8}{19} &
\cellcolor[gray]{0.8}{21} &
\cellcolor[gray]{0.8}{-2} & 4 & 1 & 3  & 0 & -1 \\
5 &
0 & 0 & -1 & 1 & 2 & 0 & -2 & -2 & -4 & 2 & 1 & \cellcolor[gray]{0.8}{9} &
\cellcolor[gray]{0.8}{-8} &
\cellcolor[gray]{0.8}{-13} &
\cellcolor[gray]{0.8}{-14} &
\cellcolor[gray]{0.8}{-18} \\
\end{tabular}
\end{center}
\caption{Financial series: Loading matrix after a varimax rotation and scaling}
\label{fin:group.loading}
\end{table}

\begin{table}
\begin{tabular}{|c|cccccc|}
\hline
Group 1 &  AssetE.R &  LiabilityE.R & & & &  \\
Group 2 & Earnings.R & EPS &  Oper.M & Profit.Margin & ROA & ROE \\
Group 3 & Cash.PS & Gross.Margin &  Revenue.PS & & & \\
Group 4 & Payout.R & Profit.G.Q & Profit.G.Y & Revenue.G.Q & Revenue.G.Y & \\ \hline
\end{tabular}
\caption{Financial series: Grouping of company financials}
 \label{fin:group}
\end{table}

Based on the 200 rows of the
estimated columns loading matrix (corresponding to the companies),
after rotation to  maximize the variance, the companies are grouped into
6 groups. Table \ref{table:comgroup} shows the grouping
corresponding to the industry classification index. The pattern is not
as clear as the clustering of the row loading matrix but we still make some interesting discoveries.  Industrial companies are mainly clustered
in Groups 1 to 3; Health Care companies in Groups 2 and 3; Information
Technology companies in Groups 1, 3 and 5; and Materials companies in
Group 3. Looking from the other angle, we find that Group 4 mainly contains Energy companies; Group 5 mainly contains Consumer Discretionary, Financials and
Information Technology companies; Group 6 mainly
contains Utility companies.

Figure \ref{CFfig5} shows the total 100 factor series in the 5 by
20 factor matrix. Interpretation of the factors is difficult. There are
significant redundancy and correlation among the factors, since we have 100
factors but the time series length is only 40.
Clearly the model tends to overfit. This example is for
illustration purpose only, though we do find interesting features.



\begin{figure}
\centerline{\includegraphics[width=5.8in,height=4.3in]{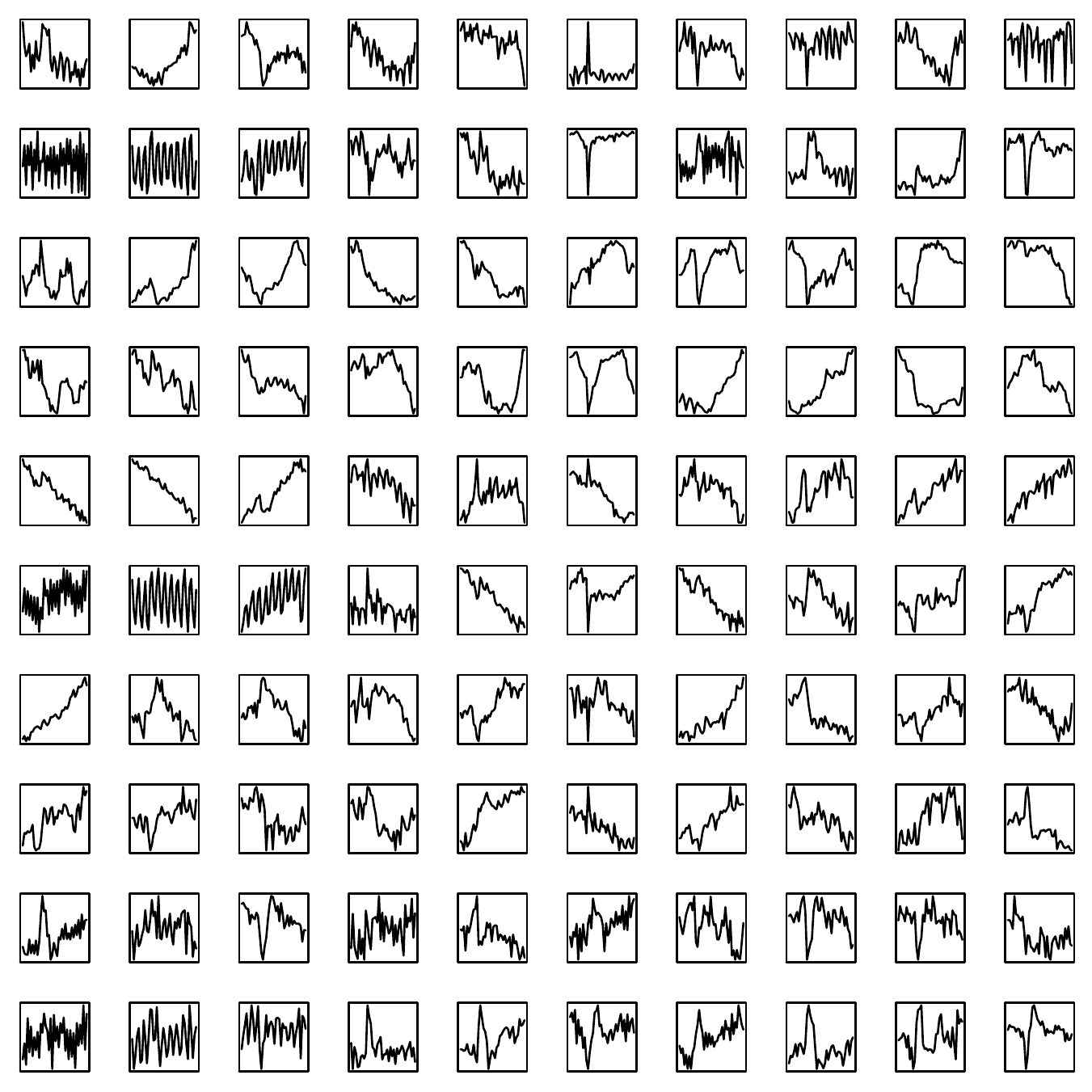}}
\caption{Financial series: Plot of the 100 series in the factor matrix}
\label{CFfig5}
\end{figure}

\begin{table}
\begin{center}
\begin{tabular}{|c|cccccc|} \hline
group &	1 &2 &3	&4 &5 &6  \\ \hline
Consumer Discretionary	& 2 & 3 & 4 & 1 & 5 & 4	\\
Consumer Staples & 3 & 4 & 6 & 1 & 0 & 1  \\
Energy	& 2 & 3 & 4 & 9 & 0 & 4  \\
Financials & 0&	5& 2 & 0& 4& 0 \\
Health Care & 0 & 5 & 17 & 0&  1 & 2 \\
Industrials & 12 & 7 & 17 & 0 &	1 & 2 \\
Information Technology & 4 & 0 & 12 & 0 & 5 & 0  \\
Materials & 2 &	5 & 8 &	1 & 1 &	1 \\
Telecommunications Services & 0 & 2 & 1 & 0 & 0 & 0 \\
Utilities & 0 & 6 & 5 & 1 & 0 & 13  \\ \hline
\end{tabular}
\end{center}
\caption{Financial series: Matching the companies and the industry}
\label{table:comgroup}
\end{table}

\begin{table}
\begin{center}
\begin{tabular}{c|ccccc}
 & factor & RSS & RSS/SST & \# factors & \# parameters \\ \hline
Matrix model & (4,10) &  86,739 & 0.701 & 40 & 2,064 \\
Matrix model & (4,20) &  74,848 & 0.610 & 80 & 4,064 \\
Matrix model & (5,10) &  84.517 & 0.688 & 50 & 2,080 \\
Matrix model & (5,20) &  71,535 & 0.582 & 100 & 4,080 \\
Matrix model & (5,30) &  65,037 & 0.530 & 150 & 6,080 \\ \hline
Vector model &    3  & 79,704  & 0.650 & 3 & 9,600 \\
Vector model &    4  & 73,457  & 0.598 & 4 & 12,800 \\
Vector model &    5  & 68,428  & 0.557 & 5 & 16,000 \\
Vector model &    6  & 63,031  & 0.514 & 6 & 19,200 \\
\end{tabular}
\end{center}
\caption{Financial series: Comparison of different models for company financials series}
\label{table:FinComparison}
\end{table}

Table \ref{table:FinComparison} shows a simple comparison between the
matrix factor models and vectorized factor models of various size and
number of factors. Since the time series is short, the table shows
in-sample residual sum of squares. Again, it is seen that
the matrix factor models use much fewer parameters in loading
matrices to achieve similar estimation performance. The number of parameters involved is large
as we are jointly modeling 3,200 time series.

\section{Summary}

In this paper we propose a matrix factor model for high-dimensional
matrix-valued time series, along with an estimation procedure. Theoretical
analysis shows the asymptotic properties of the estimators. Simulation
and real examples are used to illustrate the model and finite sample
properties of the estimators. The real examples show the usefulness of the
model and its ability to reveal interesting features of high-dimensional
time series. Significant amount of effort is needed to investigate model
validation and model comparison procedures for the proposed model. Extensions
to multi-term model and approaches to simply reducing factor redundancy
are important research topics. Extending the model to dynamic factor model
with an imposed dynamic structure on the factor matrix will be useful in
terms of prediction and better
understanding the dynamic nature of the matrix-valued time
series.

\vspace{0.2in}

\noindent
{\Large \bf Acknowledgments}

We thank the Editors and two anonymous
referees for their helpful
insightful comments and suggestions, which lead to significant improvement
of the paper in motivation and justification, design of simulation study and
the analysis of real examples.

\vspace{0.2in}

\bibliographystyle{apalike}
\bibliography{matfm1}

\newpage
\noindent
{\Large \bf Appendix 1: Proofs}

We start by defining some quantities used in the
proofs. Write
\begin{eqnarray*}
\bOmega_{s,ij}(h) &=& \frac{1}{T-h} \sum_{t=1}^{T-h} \Cov (\bR\bF_t  \bc_{i\cdot}', \bR\bF_{t+h}\bc_{j\cdot}'), \\
\bOmega_{fc,ij}(h) &=& \frac{1}{T-h} \sum_{t=1}^{T-h} \Cov (\bF_t  \bc_{i\cdot}', \bF_{t+h}\bc_{j\cdot}')\\
\widehat{\bOmega}_{s,ij}(h) &=& \frac{1}{T-h}  \sum_{t=1}^{T-h} \bR\bF_t \bc_{i\cdot}' \bc_{j\cdot} \bF_{t+h}'\bR',\\
\widehat{\bOmega}_{fc,ij}(h) &=& \frac{1}{T-h} \sum_{t=1}^{T-h} \bF_t \bc_{i\cdot}' \bc_{j\cdot} \bF_{t+h}',\\
\widehat{\bOmega}_{s\epsilon,ij}(h) &=& \frac{1}{T-h} \sum_{t=1}^{T-h}\bR \bF_t\bc_{i\cdot}' \bepsilon_{t+h,j}',\\
\widehat{\bOmega}_{\epsilon s,ij}(h)&=&\frac{1}{T-h}  \sum_{t=1}^{T-h} \bepsilon_{t,i}\bc_{j\cdot} \bF_{t+h}'\bR',\\
\widehat{\bOmega}_{\epsilon,ij}(h)&=&\frac{1}{T-h}  \sum_{t=1}^{T-h} \bepsilon_{t,i} \bepsilon_{t+h,j}'.
\end{eqnarray*}

The following lemma establishes the entry-wise convergence rate of the
covariance matrix estimation of the vectorized latent factor process
$\rmvec(\bF_t)$.

\begin{lemma} Let $F_{t,ij}$ denote the $ij$-th entry of
  $\bF_t$. Under Conditions 1 and 2, for any $i,k=1,2,\ldots,k_1$, and
  $j,l=1,2,\ldots, k_2$, it follows that
\begin{equation*}
\Big| \frac{1}{T-h} \sum_{t=1}^{T-h} \Big( F_{t,ij} F_{t+h,kl}-
       \Cov (F_{t,ij}, F_{t+h,kl})\Big) \Big|= O_p(T^{-1/2}).
\end{equation*}
\end{lemma}

\begin{proof}
Under Conditions 1 and 2, by Davydov's inequality, it follows that
\begin{eqnarray*}
\lefteqn{\expec \left( \frac{1}{T-h}\sum_{t=1}^{T-h} \Big( F_{t,ij} F_{t+h,kl} - \Cov(F_{t,ij}, F_{t+h, kl})\Big) \right)^2}\\
&=&\frac{1}{(T-h)^2} \sum_{|t_1-t_2|\leq h} {\rm E}  [F_{t_1,ij} F_{t_1+h, kl}- {\rm E} (F_{t_1,ij} F_{t_1+h,kl})]
 [F_{t_2,ij} F_{t_2+h, kl}- {\rm E} (F_{t_2,ij} F_{t_2+h,kl})]\}\\
&&+ \frac{1}{(T-h)^2} \sum_{|t_1- t_2|>h} {\rm E}  [F_{t_1,ij} F_{t_1+h, kl}- {\rm E} (F_{t_1,ij} F_{t_1+h,kl})]
 [F_{t_2,ij} F_{t_2+h, kl}- {\rm E} (F_{t_2,ij} F_{t_2+h,kl})]\} \\
&\leq &\frac{C}{T-h}+\frac{C}{T-h} \sum_{u=1}^{T-2h-1} \alpha(u)^{1-2/\gamma}=O(T^{-1}).
\end{eqnarray*}
Here $C$ denotes a constant.
\end{proof}

Under the matrix-valued factor model \eqref{eqn:2dmodel}, the
$\bR \bF_t \bC'$ can be view as the signal part and $\bE_t$ as
noise. The following lemma concerns the rates of convergence for
estimation of the signal, the noise, and the interaction between the two.
\begin{lemma}
\label{lemma:fourrates}
 Under Conditions 1-4, it holds that
\begin{eqnarray*}
\sum_{i=1}^{p_2} \sum_{j=1}^{p_2} \|\widehat{\bOmega}_{s,ij}(h)-\bOmega_{s,ij}(h)\|_2^2&=&O_p(p_1^{2-2\delta_1} p_2^{2-2\delta_2}T^{-1}),\\
\sum_{i=1}^{p_2} \sum_{j=1}^{p_2} \|\widehat{\bOmega}_{s\epsilon,ij}(h)\|^2_2&=&O_p(p_1^{2-\delta_1}p_2^{2-\delta_2}T^{-1}),\\
\sum_{i=1}^{p_2} \sum_{j=1}^{p_2} \|\widehat{\bOmega}_{\epsilon s,ij}(h)\|_2^2&=&O_p(p_1^{2-\delta_1}p_2^{2-\delta_2}T^{-1}),\\
\label{eqn:epsilonhatrate}
\sum_{i=1}^{p_2} \sum_{j=1}^{p_2} \|\widehat{\bOmega}_{\epsilon,ij}(h)\|_2^2&=&O_p(p_1^2p_2^2 T^{-1}).
\end{eqnarray*}
\end{lemma}

\begin{proof}
Firstly, we have
\begin{eqnarray*}
\lefteqn{\|\widehat{\bOmega}_{fc,ij}(h)-\bOmega_{fc,ij}(h)\|_2^2 \leq \| \widehat{\bOmega}_{fc,ij}(h) -\bOmega_{fc,ij}(h)\|_F^2} \nonumber \\
& =&\| \rmvec(\widehat{\bOmega}_{fc,ij}(h)-{\bOmega}_{fc,ij}(h))\|_2^2 \nonumber \\
&=&\Big\| \frac{1}{T-h}\sum_{t=1}^{T-h}\rmvec (\bF_{t}\bc_{i\cdot}' \bc_{j \cdot}  \bF_{t+h}'- \expec (\bF_{t} \bc_{i\cdot}' \bc_{j \cdot} \bF_{t+h}) ) \Big\|_2^2 \nonumber \\
&=&\Big\| \frac{1}{T-h}\sum_{t=1}^{T-h} \big[\bF_{t+h} \otimes \bF_t- \expec (\bF_{t+h} \otimes \bF_t) \big]\rmvec (\bc_{i\cdot}' \bc_{j\cdot})\Big\|_2^2 \nonumber \\
&\leq &  \Big\| \frac{1}{T-h}\sum_{t=1}^{T-h}  (\bF_{t+h} \otimes \bF_t- \expec (\bF_{t+h} \otimes \bF_t)\Big\|_2^2  \| \rmvec (\bc_{i\cdot}'\bc_{j\cdot})\|_2^2 \\
&=&  \Big\| \frac{1}{T-h}\sum_{t=1}^{T-h}  (\bF_{t+h} \otimes \bF_t- \expec (\bF_{t+h} \otimes \bF_t)\Big\|_F^2  \| \bc_{i\cdot}' \bc_{j\cdot}\|_F^2 \\
&\leq & \Big\| \frac{1}{T-h}\sum_{t=1}^{T-h}  (\bF_{t+h} \otimes \bF_t- \expec (\bF_{t+h} \otimes \bF_t)\Big\|_F^2  \|\bc_{i\cdot}\|_2^2 \cdot \|\bc_{j\cdot}\|_2^2.
\end{eqnarray*}
Hence, by Condition 4 and Lemma 1, it follows that
\begin{eqnarray*}
\lefteqn{\sum_{i=1}^{p_2} \sum_{j=1}^{p_2}\| \widehat{\bOmega}_{s,ij}(h)-\bOmega_{s,ij}(h)\|_2^2=\sum_{i=1}^{p_2} \sum_{j=1}^{p_2}\| \bR (\widehat{\bOmega}_{fc,ij}(h)-\bOmega_{fc,ij}(h))\bR'\|_2^2}\\
&\leq&  \|\bR\|_2^4 \Big\| \frac{1}{T-h}\sum_{t=1}^{T-h}  (\bF_{t+h} \otimes \bF_t- \expec (\bF_{t+h} \otimes \bF_t)\Big\|_F^2 \left( \sum_{i=1}^{p_2}\|\bc_{i\cdot}\|_2^2 \right)^2 \\
&=&\|\bR\|_2^4 \Big\| \frac{1}{T-h}\sum_{t=1}^{T-h}  (\bF_{t+h} \otimes \bF_t- \expec (\bF_{t+h} \otimes \bF_t)\Big\|_F^2 \|\bC\|_F^4\\
&\leq & k_2^2 \|\bR\|_2^4 \Big\| \frac{1}{T-h}\sum_{t=1}^{T-h}  (\bF_{t+h} \otimes \bF_t- \expec (\bF_{t+h} \otimes \bF_t)\Big\|_F^2 \|\bC\|_2^4=O_p(p_1^{2-2\delta_1}p_2^{2-2\delta_2}T^{-1}).
\end{eqnarray*}

Similarly, for the interaction component between signal and noise, we have
\begin{eqnarray*}
\lefteqn{\sum_{i=1}^{p_2} \sum_{j=1}^{p_2}\| \widehat{\bOmega}_{s\epsilon,ij}(h)\|_2^2 \leq  \sum_{i=1}^{p_2}\sum_{j=1}^{p_2} \|\bR\|_2^2 \Big\| \frac{1}{T-h} \sum_{t=1}^{T-h} \bF_t\bc_{i\cdot}'\bepsilon_{t+h,j}' \Big\|_2^2}\\
&\leq&  \|\bR\|_2^2  \left(\sum_{j=1}^{p_2} \Big\| \frac{1}{T-h} \sum_{t=1}^{T-h} \bepsilon_{t+h,j} \otimes \bF_t  \Big\|_2^2 \right) \left( \sum_{i=1}^{p_2} \|  \bc_{i\cdot}\|_2^2 \right)\\
&= &O_p(p_1^{2-\delta_1}p_2^{2-\delta_2}T^{-1}),
\end{eqnarray*}
and
\begin{eqnarray*}
\sum_{i=1}^{p_2} \sum_{j=1}^{p_2} \|\widehat{\bOmega}_{\epsilon s,ij}(h)\|_2^2&=&O_p(p_1^{2-\delta_1}p_2^{2-\delta_2}T^{-1}).
\end{eqnarray*}

Lastly, for the noise term, we have
\begin{eqnarray*}
\sum_{i=1}^{p_2} \sum_{j=1}^{p_2}\|\widehat{\bOmega}_{\epsilon,ij}(h)\|_2^2= \sum_{i=1}^{p_2} \sum_{j=1}^{p_2} \|\frac{1}{T-h} \sum_{t=1}^{T-h} \bepsilon_{t,i} \bepsilon_{t+h,j}' \|_2^2
=O_p(p_1^2p_2^2 T^{-1}).
\end{eqnarray*}
\end{proof}

With the four rates established in Lemma \ref{lemma:fourrates}, we
can now study the rate of convergence for the observed covariance
matrix $\widehat{\bOmega}_{x,ij}(h)$.
\begin{lemma}
\label{lemma:omegayijdiffrate}
Under Conditions 1-4, it holds that
\begin{equation*}
\sum_{i=1}^{p_2} \sum_{j=1}^{p_2} \|\widehat{\bOmega}_{x,ij}(h)- \bOmega_{x,ij}(h)\|_2^2=O_p(p_1^2 p_2^2 T^{-1}).
\end{equation*}
\end{lemma}

\begin{proof}
From the definition of $\widehat{\bOmega}_{x,ij}(h)$
in \eqref{eqn:hatomegayijdef}, we can decompose
$\widehat{\bOmega}_{x,ij}(h)$ into four parts as follows,
\begin{eqnarray*}
\lefteqn{\widehat{\bOmega}_{x,ij}(h)=\frac{1}{T-h}\sum_{t=1}^{T-h} \bx_{t,i}\bx_{t+h,j}'}\\
&=&\frac{1}{T-h}\sum_{t=1}^{T-h} (\bR \bF_t \bc_{i\cdot}'+\bepsilon_{t,i})(\bR \bF_{t+h} \bc_{j\cdot}'+\bepsilon_{t+h,j})'\\
&=& \widehat{\bOmega}_{s,ij}(h)+ \widehat{ \bOmega}_{s\epsilon,ij} (h) +\widehat{\bOmega}_{\epsilon s,ij}(h)+ \widehat{\bOmega}_{\epsilon,ij}(h).
\end{eqnarray*}
Then by Lemma 2, it follows that
\begin{eqnarray*}
\lefteqn{\sum_{i=1}^{p_2} \sum_{j=1}^{p_2} \big\|  \widehat{\bOmega}_{x,ij}(h)- \bOmega_{x,ij}(h) \big\|_2^2}\nonumber \\
&\leq& 4 \sum_{i=1}^{p_2} \sum_{j=1}^{p_2}\left( \| \widehat{ \bOmega}_{s,ij}(h)-\bOmega_{s,ij}(h) \|_2^2 +\|\widehat{\bOmega}_{s\epsilon,ij}(h)\|_2^2
+\| \widehat{\bOmega}_{\epsilon s,ij}(h)\|_2^2+\|\widehat{\bOmega}_{\epsilon,ij}(h)\|_2^2 \right) \nonumber\\
&=&O_p(p_1^2p_2^2 T^{-1}).
\end{eqnarray*}
\end{proof}

\begin{lemma} Under Conditions 1-4, and $p_1^{\delta_1}p_2^{\delta_2} T^{-1/2}=o(1)$, it holds that
\[
\|\widehat{\bM}_1 -\bM_1\|_2= O_p(p_1^{2-\delta_1} p_2^{2-\delta_2} T^{-1/2}).
\]
\end{lemma}
\begin{proof}
From the definitions of $\widehat{\bM}_1$ and $\bM_1$ in \eqref{eqn:hatmdef} and \eqref{eqn:Mexpression}, it follows that
\begin{eqnarray*}
\lefteqn{\|\widehat{\bM}_1-\bM_1\|_2 =\Big\| \sum_{h=1}^{h_0} \sum_{i=1}^{p_2}
\sum_{j=1}^{p_2}
\Big(\widehat{\bOmega}_{x,ij}(h) \widehat{\bOmega}_{x,ij}'(h)
- \bOmega_{x,ij}(h) \bOmega_{x,ij}'(h) \Big) \Big\|_2}\\
&\leq& \sum_{h=1}^{h_0}  \sum_{i=1}^{p_2}\sum_{j=1}^{p_2} \Big( \| ( \widehat{\bOmega}_{x,ij}(h)-\bOmega_{x,ij}(h) ) ( \widehat{\bOmega}_{x,ij}(h)-\bOmega_{x,ij}(h) )' \|_2
+ 2 \|\bOmega_{x,ij}(h)\|_2  \|\widehat{\bOmega}_{x,ij}(h)-\bOmega_{x,ij}(h) \|_2 \Big ) \\
&\leq&  \sum_{h=1}^{h_0} \sum_{i=1}^{p_2} \sum_{j=1}^{p_2} \|  \widehat{\bOmega}_{x,ij}(h)-\bOmega_{x,ij}(h) \|_2^2
+2 \sum_{h=1}^{h_0}  \sum_{i=1}^{p_2}\sum_{j=1}^{p_2} \|\bOmega_{x,ij}(h)\|_2  \|\widehat{\bOmega}_{x,ij}(h)-\bOmega_{x,ij}(h) \|_2.
\end{eqnarray*}

We have
\begin{eqnarray}
\lefteqn{\sum_{i=1}^{p_2} \sum_{j=1}^{p_2}\|\bOmega_{x,ij}(h)\|_2^2 =\sum_{i=1}^{p_2} \sum_{j=1}^{p_2} \| \bR\bOmega_{fc,ij}(h) \bR'\|_2^2
\leq \sum_{i=1}^{p_2} \sum_{j=1}^{p_2} \|\bR\|_2^4 \|\bOmega_{fc,ij}(h)\|_2^2} \nonumber \\
&  \leq & \|\bR\|_2^4 \cdot \Big\| \frac{1}{T-h}\sum_{t=1}^{T-h}  \expec (\bF_{t+h} \otimes \bF_t)\Big\|_2^2 \cdot   \left(\sum_{i=1}^{p_2}\| \bc_{i\cdot}\|_2^2\right)^2 \nonumber \\
& = & O(p_1^{2-2\delta_1}p_2^{2-2\delta_2}). \label{eqn:gamma1}
\end{eqnarray}

By \eqref{eqn:gamma1} and Lemma \ref{lemma:omegayijdiffrate},
\begin{eqnarray}
\lefteqn{\left(\sum_{i=1}^{p_2}\sum_{j=1}^{p_2} \|\bOmega_{x,ij}(h)\|_2  \|\widehat{\bOmega}_{x,ij}(h)-\bOmega_{x,ij}(h) \|_2\right)^2} \nonumber \\
&\leq& \left(\sum_{i=1}^{p_2}\sum_{j=1}^{p_2} \|\bOmega_{x,ij}(h)\|_2^2\right) \cdot \left( \sum_{i=1}^{p_2}\sum_{j=1}^{p_2}  \|\widehat{\bOmega}_{x,ij}(h)-\bOmega_{x,ij}(h) \|_2^2  \right)\nonumber \\
\label{eqn:productrate}
&\leq& O_p(p_1^{2-2\delta_1}p_2^{2-2\delta_2}p_1^2p_2^2T^{-1}) = O_p(p_1^{4-2\delta_1}p_2^{4-2\delta_2}T^{-1}),
\end{eqnarray}
where the second inequality follows from Cauchy-Schwarz inequality.

From \eqref{eqn:productrate}, Lemma \ref{lemma:omegayijdiffrate}, and the
condition $p_1^{\delta_1}p_2^{\delta_2} T^{-1/2}=o(1)$, it follows
that
\[
\|\widehat{\bM}_1-\bM_1\|_2 =O_p(p_1^{2-\delta_1} p_2^{2-\delta_2} T^{-1/2}).
\]
\end{proof}

\begin{lemma} Under Conditions 2 and 3, we have
\[
\lambda_i(\bM_1) \asymp p_1^{2-2 \delta_1} p_2^{2-2\delta_2}, \quad i=1,2,... k_1,
\]
where $\lambda_i(\bM_1)$ denotes the $i$-th largest eigenvalue of $\bM_1$.
\end{lemma}
\begin{proof}
By definition, we have
\begin{eqnarray*}
\bOmega_{fc,ij}(h)&=&\frac{1}{T-h}\sum_{t=1}^{T-h} {\rm E} \big[(\bc_{i\cdot} \otimes \bI_{k_1}){\rm vec}(\bF_t) {\rm vec}( \bF_{t+h})'( \bc_{j \cdot}' \otimes \bI_{k_1}) \big]\\
&=& (\bc_{i \cdot} \otimes \bI_{k_1})\bSigma_{f}(h)( \bc_{j\cdot}' \otimes \bI_{k_1}).
\end{eqnarray*}

Under Conditions 2-3 and by properties of Kronecker product we have
\begin{eqnarray*}
\lambda_{k_1}(\bM_1)&=&\lambda_{k_1} \left(\sum_{h=1}^{h_0} \sum_{i=1}^{p_2} \sum_{j=1}^{p_2} \bR \bOmega_{fc,ij}(h)\bR' \bR \bOmega_{fc,ij}'(h) \bR' \right)\\
&\geq& \| \bR\|_{\min}^4 \cdot   \lambda_{k_1} \left( \sum_{h=1}^{h_0} \sum_{i=1}^{p_2} \sum_{j=1}^{p_2}\bOmega_{fc,ij}(h) {\bOmega}_{fc,ij}'(h)\right)\\
&=& \|\bR\|_{\min}^4  \cdot  \lambda_{k_1} \left( \sum_{h=1}^{h_0} \sum_{i=1}^{p_2} \sum_{j=1}^{p_2}   (\bc_{i\cdot} \otimes \bI_{k_1}) \bSigma_f(h) ( \bc_{j}' \otimes \bI_{k_1}) (\bc_{j\cdot} \otimes \bI_{k_1})\bSigma_f'(h) (\bc_{i\cdot}' \otimes \bI_{k_1}) \right)\\
&\geq& \|\bR\|_{\min}^4  \cdot   \lambda_{k_1} \left(\sum_{h=1}^{h_0} \sum_{i=1}^{p_2} \sum_{j=1}^{p_2} (\bc_{i\cdot} \otimes \bI_{k_1}) \bSigma_f(h) (\bc_{j\cdot}' \bc_{j\cdot}  \otimes \bI_{k_1})\bSigma_f'(h) (\bc_{i\cdot}' \otimes \bI_{k_1}) \right)\\
&=& \|\bR\|_{\min}^4  \cdot   \lambda_{k_1} \left(\sum_{h=1}^{h_0} \sum_{i=1}^{p_2}  (\bc_{i\cdot} \otimes \bI_{k_1}) \bSigma_f(h) (\bC' \bC \otimes  \bI_{k_1} )\bSigma_f'(h) (\bc_{i\cdot}' \otimes \bI_{k_1}) \right)\\
&=& \|\bR\|_{\min}^4  \cdot   \lambda_{k_1} \left(\sum_{h=1}^{h_0} \sum_{i=1}^{p_2}  (\bc_{i\cdot} \otimes \bI_{k_1}) \bSigma_f(h) (\bC' \otimes \bI_{k_1}  ) (\bC \otimes \bI_{k_1} )\bSigma_f'(h) (\bc_{i\cdot}' \otimes \bI_{k_1}) \right)\\
&=& \|\bR\|_{\min}^4  \cdot   \lambda_{k_1} \left(\sum_{h=1}^{h_0} \sum_{i=1}^{p_2}  (\bC \otimes \bI_{k_1} ) \bSigma_f'(h) (\bc_{i\cdot}'\otimes \bI_{k_1}) (\bc_{i\cdot} \otimes \bI_{k_1}) \bSigma_f(h) (\bC' \otimes  \bI_{k_1} )\right)\\
&=& \|\bR\|_{\min}^4  \cdot   \lambda_{k_1} \left(\sum_{h=1}^{h_0}   (\bC \otimes \bI_{k_1}  ) \bSigma_f'(h) (\bC'\bC \otimes \bI_{k_1})\bSigma_f(h)  (\bC' \otimes \bI_{k_1}  )\right).
\end{eqnarray*}
Since $\bC'\bC$ is a $k_2\times k_2$ symmetric positive definite matrix, we can find a $k_2\times k_2$ positive definite matrix $\bU$, such that $\bC'\bC=\bU\bU'$ and $\|\bU\|_2^2\asymp   O(p_2^{1-\delta_2}) \asymp\|\bU\|_{\min}^2$. By the properties of Kronecker product, we can show that $\sigma_{1}(\bU \otimes \bI_{k_1})\asymp O(p_2^{1/2-\delta_2/2}) \asymp \sigma_{k_1k_2}(\bU \otimes \bI_{k_1})$. Under Condition 2 using Theorem 9 in \cite{merikoski2004}, it follows that
$\sigma_{k_1}\left(\bSigma_{f}'(h)(\bU \otimes \bI_{k_1})\right) \asymp O(p_2^{1/2-\delta_2/2})$.

Using Theorem 9 in \cite{merikoski2004} again, we have
\begin{eqnarray*}
\lambda_{k_1}(\bM_1) &\geq & \|\bR\|_{\min}^4  \cdot   \lambda_{k_1} \left(\sum_{h=1}^{h_0}   (\bC \otimes \bI_{k_1}  ) \bSigma_f'(h) (\bU \otimes \bI_{k_1})  (\bU' \otimes \bI_{k_1})\bSigma_f(h)  (\bC' \otimes \bI_{k_1}  )\right)\\
&=& \|\bR\|_{\min}^4  \cdot   \lambda_{k_1} \left(\sum_{h=1}^{h_0}    (\bU' \otimes \bI_{k_1})\bSigma_f(h)  (\bC' \bC \otimes \bI_{k_1}  ) \bSigma_f'(h) (\bU \otimes \bI_{k_1}) \right)\\
&=& \|\bR\|_{\min}^4  \cdot   \lambda_{k_1} \left(\sum_{h=1}^{h_0}    (\bU' \otimes \bI_{k_1})\bSigma_f(h) (\bU \otimes \bI_{k_1})(\bU' \otimes \bI_{k_1})\bSigma_f'(h) (\bU \otimes \bI_{k_1}) \right)\\
&\geq& \|\bR\|_{\min}^4  \cdot   \left[ \sigma_{k_1}\left((\bU' \otimes \bI_{k_1}) \bSigma_f'(h) (\bU \otimes \bI_{k_1}) \right) \right]^2 = O(p_1^{2-2\delta_1}p_2^{2-2\delta_2}).
\end{eqnarray*}
\end{proof}
\noindent{\bf Proof of Theorem 1}
\begin{proof}
 By Lemmas 1-5, and Lemma 3 in \citet{lam2011estimation}, Theorem 1 follows.
\end{proof}

\noindent{\bf Proof of Theorem 2}
\begin{proof}
The proof is quite similar to that of Theorem 1 of
\citet{lam2012factor}.  We denote $\widehat{\lambda}_{1,j}$ and
$\widehat{\bq}_{1,j}$ for the $j$-th largest eigenvalue of
$\widehat{\bM}_1$ and its corresponding eigenvector, respectively. The
corresponding population eigenvalues are denoted by $\lambda_{1,j}$ and
$\bq_{1,j}$ for the matrix $\bM_1$.  Let
$\widehat{\bQ}_1=(\widehat{\bq}_{1,1}, \ldots, \widehat{\bq}_{1,k_1})$ and
$\bQ_1=(\bq_{1,1}, \ldots, \bq_{1,k_1})$. We have
\[
\lambda_{1,j}=\bq_{1,j}' \bM_1 \bq_{1,j}, \mbox{\ \ and \ \ }
\quad \widehat{\lambda}_{1,j}=\widehat{\bq}_{1,j}' \widehat{\bM}_1 \widehat{\bq}_{1,j}, \qquad j=1, \ldots, p_1.
\]
We can decompose $\widehat{\lambda}_{1,j}-\lambda_{1,j}$ by
\[
\widehat{\lambda}_{1,j}-\lambda_{1,j}= \widehat{\bq}_{1,j}' \widehat{\bM}_1\, \widehat{\bq}_{1,j} -\bq_{1,j}' \bM_1 \bq_{1,j}=I_1 +I_2 +I_3 +I_4 +I_5,
\]
where
\[
I_1= (\widehat{\bq}_{1,j}-\bq_{1,j})' (\widehat{\bM}_1-\bM_1)\widehat{\bq}_{1,j}, \quad I_2=(\widehat{\bq}_{1,j}-\bq_{1,j})' \bM_1 (\widehat{\bq}_{1,j}-\bq_{1,j}),
\]
\[
I_3=(\widehat{\bq}_{1,j}-\bq_{1,j})' \bM_1\bq_{1,j}, \quad I_4=\bq_{1,j}' (\widehat{\bM}_1-\bM_1)\widehat{\bq}_{1,j}, \quad
I_5= \bq_{1,j}' \bM_1  (\widehat{\bq}_{1,j}-\bq_{1,j}).\]
For $j=1,\ldots, k_1$, $\|\widehat{\bq}_{1,j}-\bq_{1,j}\|_2 \leq \|\widehat{\bQ}_1-\bQ_1\|_2=O_p(h_{T})$, where $h_{T}=p_1^{\delta_1}p_2^{\delta_2}T^{-1/2}$ by Theorem 1, and $\|\bM_1\|_2=O_p(p_1^{2-\delta_1}p_2^{2-\delta_2})$. By Lemma 4, we have
$\|I_1\|_2$ and $\|I_2\|_2$ are of order
$O_p(p_1^{2-2\delta_1}p_2^{2-2\delta_2}h_{T}^2)$ and $\|I_3\|_2$, $\|I_4\|_2$ and
$\|I_5\|_2$ are of order $O_p(p_1^{2-2\delta_1}p_2^{2-2\delta_2}h_{T})$.
So $|\widehat{\lambda}_{1,j}-\lambda_{1,j}|=O_p(p_1^{2-2\delta_1}p_2^{2-2\delta_2}h_{T})=O_p(p_1^{2-\delta_1}p_2^{2-\delta_2}T^{-1/2})$.

\noindent For $j=k_1+1, \ldots, p_1$, define,
\[
\widetilde{\bM}_1=\sum_{h=1}^{h_0} \sum_{i=1}^{p_2} \sum_{j=1}^{p_2} \widehat{\bOmega}_{x,ij}(h) \bOmega_{x,ij}'(h), \quad \widehat{\bB}_1=(\widehat{\bq}_{1,k_1+1}, \ldots, \widehat{\bq}_{1,p_1}),
\mbox{\ \ and \ \ } \bB_1=(\bq_{1,k_1+1}, \ldots, \bq_{1,p_1}).
\]
It can be shown that $\|\widehat{\bB}_1-\bB_1\|_2=O_p(h_T)$, similar to proof of Theorem 1 with Lemma 3 in \citet{lam2011estimation}. Hence,
$\|\widehat{\bq}_{1,j}-\bq_{1,j}\|_2 \leq \|\widehat{\bB}_1-\bB_1\|_2=O_p(h_{T})$.\\
Since $\lambda_{1,j}=0$, for $j=k_1+1, \ldots, p_1$, consider the decomposition
\[
\widehat{\lambda}_{1,j}= \widehat{\bq}_{1,j}' \widehat{\bM}_1 \widehat{\bq}_{1,j}=K_1+K_2+K_3,
\]
where
\[
K_1= \widehat{\bq}_{1,j}' (\widehat{\bM}_1-\widetilde{\bM}_1-\widetilde{\bM}_1+\bM_1) \widehat{\bq}_{1,j}, \quad
K_2= 2 \widehat{\bq}_{1,j}' (\widetilde{\bM}_1-\bM_1)(\widehat{\bq}_{1,j}-\bq_{1,j}),
\]
\[
K_3=( \widehat{\bq}_{1,j}-\bq_{1,j})' \bM_1(\widehat{\bq}_{1,j}-\bq_{1,j}).
\]
By Lemma 2 and Lemma 4,
\begin{eqnarray*}
K_1 & = &  \sum_{h=1}^{h_0} \|\sum_{i=1}^{p_2} \sum_{j=1}^{p_2}(\widehat{\bOmega}_{x,ij}(h) -\bOmega_{x,ij}(h))\widehat{\bq}_{1,j}\|_2^2 \leq \sum_{h=1}^{h_0} \sum_{i=1}^{p_2} \sum_{j=1}^{p_2} \|\widehat{\bOmega}_{x,ij}(h) -\bOmega_{x,ij}(h)\|_2^2=O_p(p_1^2p_2^2 T^{-1}), \\
|K_2| & = &
O_p(\|\widetilde{\bM}_1-\bM_1\|_2 \cdot \|\widehat{\bq}_{1,j}-\bq_{1,j} \|_2)=O_p(\|\widetilde{\bM}-\bM_1\|_2 \cdot \|\widehat{\bB}_1-\bB_1\|_2)=O_p(p_1^2p_2^2T^{-1}), \\
|K_3|& = & O_p(\|\widehat{\bB}_1-\bB_1\|_2^2 \cdot \|\bM_1\|_2)=O_p(p_1^{2-2\delta_1}p_2^{2-2\delta_2}h_T^2)=O_p(p_1^2p_2^2T^{-1}).
\end{eqnarray*}
Hence $\widehat{\lambda}_{1,j}=O_p(p_1^2p_2^2T^{-1})$.

If we use the transpose of $\bX_t$ to construct $\bM_2$, we can obtain the asymptotic properties of the eigenvalues of estimated $\bM_2$ in a similar way.
\end{proof}

\noindent{\bf Proof of Theorem 3}
\begin{proof}
\begin{align*}
\widehat{\bS}_t-\bS_t=&  \widehat{\bQ}_1 \widehat{\bQ}_1' \bX_t \widehat{\bQ}_2 \widehat{\bQ}_2'-\bQ_1 \bZ_t \bQ_2'= \widehat{\bQ}_1 \widehat{\bQ}_1' (\bQ_1 \bZ_t\bQ_2' +\bE_t)\widehat{\bQ}_2 \widehat{\bQ}_2'-\bQ_1\bQ_1' \bQ_1 \bZ_t \bQ_2' \bQ_2\bQ_2'\\
=&\widehat{\bQ}_1 \widehat{\bQ}_1' \bQ_1\bZ_t \bQ_2'(\widehat{\bQ}_2\widehat{\bQ}_2'-\bQ_2\bQ_2') +(\widehat{\bQ}_1\widehat{\bQ}_1' - \bQ_1\bQ_1') \bQ_1 \bZ_t \bQ_2' + \widehat{\bQ}_1 \widehat{\bQ}_1' \bE_t \widehat{\bQ}_2 \widehat{\bQ}_2' \\
=&I_1+I_2 +I_3 .
\end{align*}
By Theorem 1, we have
\begin{align*}
\|I_1\|_2&\leq 2\|\bZ_t\|_2 \| \widehat{\bQ}_2 -\bQ_2\|_2= O_p(p_1^{1/2-\delta_1/2} p_2^{1/2-\delta_2/2}\| \widehat{\bQ}_2 -\bQ_2\|_2)=O_p(p_1^{1/2+\delta_1/2}p_2^{1/2+\delta_2/2}T^{-1/2}),\\
\|I_2\|_2&\leq 2\| \widehat{\bQ}_1 -\bQ_1\|_2 \|\bZ_t\|_2 =O_p(p_1^{1/2-\delta_1/2} p_2^{1/2-\delta_2/2}\| \widehat{\bQ}_1 -\bQ_1\|_2)=O_p(p_1^{1/2+\delta_1/2}p_2^{1/2+\delta_2/2}T^{-1/2}),\\
\|I_3\|_2&\leq \| \widehat{\bQ}_1' \bE_t \widehat{\bQ}_2\|_2= \| (\widehat{\bQ}_2' \otimes \widehat{\bQ}_1'){\rm vec}(\bE_t)\|_2 \leq k_1k_2 \| \bSigma_e\|_2 =O_p(1).
\end{align*}
The conclusion follows.
\end{proof}

\noindent{\bf Proof of Theorem 4}
\begin{proof}
We assume that $\bQ_1$ is uniquely defined as $\bQ_1=(\bq_{1,1}, \bq_{1,2}, \ldots, \bq_{1,k_1})$, where $\bq_{1,1}, \ldots, \bq_{1,k_1}$ are eigenvectors of $\bM_1$ corresponding to the largest $k_1$ eigenvalues $\lambda_{1,1}, \ldots, \lambda_{1,k_1}$, and $\lambda_{1,1}> \lambda_{1,2} > \ldots > \lambda_{1,k_1}$. Then similar to proof of Theorem 3 in \citet{liu2016regime}, we can obtain the results.
\end{proof}

\vspace{0.5in}

\noindent
{\Large \bf Appendix 2: Definitions of Financials Used}

The following table shows the definition of the company financials used
in the analysis. Some are directly reported by the company in their
quarterly reports, and some are derived using the reported figures.

\begin{center}
\begin{table}
\begin{tabular}{lll}
Short Name & Variable Name & Calculation  \\ \hline
Profit.M & Profit Margin &  Net Income/Revenue  \\	
Oper.M & Operating Margin & Operating Income / Revenue  \\	
EPS & Diluted Earing per share	& from report \\	
Gross.Margin & Gross Margin	& Gross Profit / Revenue \\	
ROE & Return on equity	& Net Income / Shareholders Equity	\\
ROA & Return on assets	& Net Income / Total Assets	\\
Revenue.PS & Revenue Per Share	& Revenue / Shares Outstanding \\
LiabilityE.R & Liability/Equity Ratio &	
Total Liabilities / Shareholders Equity	\\
AssetE.R & Asset/Equity Ratio &	Total Assets / Shareholders Equity \\
Earnings.R & Basic Earnings Power Ratio	& EBIT / Total Assets	\\
Payout.R & Payout Ratio	& Dividend Per Share / EPS Basic \\
Cash.PS & Cash Per Share & Cash and other / Shares Outstanding \\
Revenue.G.Q & Revenue Growth over last Quarter & Revenue/ Revenue Last Quarter $-1$ \\
Revenue.G.Y & Revenue Growth over same Quarter Last Year & Revenue/ Revenue Last Year $-1$ \\	
Profit.G.Q & Profit Growth over last Quarter & Profit / Profit Last Quarter
$-1$
\\	
Profit.G.Y & Profit Growth over same Quarter last Year &
profit / Profit Last Quarter $-1$ \\ \hline
\end{tabular}
\end{table}
\end{center}

In calculating profit growth ratio, an NA is recorded when profit changes from
negative to positive or from positive to negative.

\end{document}